\documentclass[12pt, draftclsnofoot, onecolumn]{IEEEtran}

\usepackage{graphicx}
\usepackage{mathrsfs}
\usepackage{stfloats}
\usepackage{dsfont}
\usepackage{setspace}
\usepackage{array}
\usepackage{bbm}
\usepackage{hyperref}
\usepackage{here}
\usepackage{amsmath,amsthm,amsfonts,amssymb}
\usepackage{color}
\usepackage{accents}
\usepackage{caption}
\usepackage{subcaption}

\usepackage{epstopdf}
\usepackage[noadjust]{cite}

\newtheorem{theorem}{Theorem}
\newtheorem{corollary}{Corollary}
\newtheorem{proposition}{Proposition}
\newtheorem{lemma}{Lemma}
\newtheorem{remark}{Remark}

\def\sir{\operatorname{SIR}}
\def\sinr{\operatorname{SINR}}
\def\asnr{\operatorname{ASNR}}
\def\pr{\operatorname{\mathbb{P}}}
\def\e{\operatorname{\mathbb{E}}}
\def\mt{\operatorname{MT}}
\def\bs{\operatorname{BS}}
\def\lambdaBS{\operatorname{\lambda_{BS}}}
\def\lambdaMT{\operatorname{\lambda_{MT}}}

\def\pcovcon{{\operatorname{P_{cov}^{\Phi}}}}

\def\PhiBS{\operatorname{\Phi_{BS}}}
\def\PhiMT{\operatorname{\Phi_{MT}}}
\def\PhiI{\operatorname{\Phi_{I}}}

\def\p_active{\operatorname{p_{active}}}

\def\r'{\bar{r}}

\newcommand\mydef{\mathrel{\stackrel{\makebox[0pt]{\mbox{\tiny$\Delta$}}}{=}}}
\newcommand\mystepA{\mathrel{\stackrel{\makebox[0pt]{\mbox{$(a)$}}}{=}}}
\newcommand\mystepB{\mathrel{\stackrel{\makebox[0pt]{\mbox{$(b)$}}}{=}}}
\newcommand\mystepC{\mathrel{\stackrel{\makebox[0pt]{\mbox{$(c)$}}}{=}}}
\newcommand\mystepD{\mathrel{\stackrel{\makebox[0pt]{\mbox{$(d)$}}}{=}}}
\newcommand\mystepE{\mathrel{\stackrel{\makebox[0pt]{\mbox{$(e)$}}}{=}}}

\newcommand\iffB{\mathrel{\stackrel{\makebox[0pt]{\mbox{$(b)$}}}{\iff}}}
\newcommand\iffC{\mathrel{\stackrel{\makebox[0pt]{\mbox{$(c)$}}}{\iff}}}

\begin{document}
	
	\title{\Huge Analysis of the Delay Distribution  in Cellular Networks by Using Stochastic Geometry}
	
	\author{\normalsize Fadil~Habibi~Danufane and \normalsize Marco~Di~Renzo,~\IEEEmembership{\normalsize Fellow,~IEEE}
		
%		\thanks{Manuscript submitted November 7, 2021. The authors are with Universit\'e Paris-Saclay, CNRS, CentraleSup\'elec, Laboratoire des Signaux et Syst\`emes, 3 Rue Joliot-Curie, 91192 Gif-sur-Yvette, France. (e-mail: fadil.danufane@centralesupelec.fr, marco.di-renzo@universite-paris-saclay.fr).}

}
	
	\maketitle
	
	\vspace{-2cm}
	
	\begin{abstract}
		In this paper, with the aid of the mathematical tool of stochastic geometry, we introduce analytical and computational frameworks for the distribution of three different definitions of delay, i.e., the time that it takes for a user to successfully receive a data packet, in large-scale cellular networks. We also provide an asymptotic analysis of one of the delay distributions, which can be regarded as the packet loss probability of a given network. To mitigate the inherent computational difficulties of the obtained analytical formulations in some cases, we propose efficient numerical approximations based on the numerical inversion method, the Riemann sum, and the Beta distribution. Finally, we demonstrate the accuracy of the obtained analytical formulations and the corresponding approximations against Monte Carlo simulation results, and unveil insights on the delay performance with respect to several design parameters, such as the decoding threshold, the transmit power, and the deployment density of the base stations. The proposed methods can facilitate the analysis and optimization of cellular networks subject to reliability constraints on the network packet delay that are not restricted to the local (average) delay, e.g., in the context of delay sensitive applications.
	\end{abstract}


	\section{Introduction}\label{sec:Introduction}
	
	\subsection{Motivation and Related Works} 
	
	The fifth generation (5G) of wireless networks  introduces ultra reliable low latency communication (URLLC) as one of its main use cases which promises to provide low latency and ultra-high reliability for mission critical applications, such as the industrial Internet, the smart grid, remote surgery, and intelligent transportation systems. Motivated by this, delay and reliability emerge as important performance metrics to consider, in addition to traditional key performance indicators, such as the capacity and the spectral efficiency, in communication systems design. According to the 3rd Generation Partnership Project (3GPP) \cite{3gpp2016study}, the end-to-end delay of 5G communications is required to be less than 1 ms, which is about 1/200 of the delay requirement for the fourth generation (4G), and the target reliability must be as high as $1 - 10^{-2}$ to $1 - 10^{-7}$.
	
	In cellular networks, the delay is tightly related to the transmission success probability (i.e., the coverage probability) which is typically determined by the signal-to-interference ratio (SIR) or the signal-to-noise-plus-interference ratio (SINR) at the user side.
	Recently, to take into account the strong interplay between the transmit power and the density of the base stations (BSs) for optimal cellular networks planning, a new definition of the coverage was introduced by the authors of \cite{di2018system}, which depends on the SIR and the average signal-to-noise ratio (ASNR).
	Since the SIR and SINR are functions of the received signal as well as the interference from the BSs in the network, they strongly depend on the inter-distances among the BSs and the users. However, characterizing their exact expressions is challenging due to the complex nature of the deployment of BSs and users locations.
	Leveraging the theory of stochastic geometry, fundamental works such as \cite{haenggi2009stochastic} and \cite{andrews2011tractable} show that many performance metrics, such as the coverage and rate can be derived in a tractable way by modeling the spatial distribution of transmitters and receivers as homogeneous Poisson point processes (PPPs). Thanks to its simplicity and tractability, this modeling approach has attracted a lot of interest and, as a result, stochastic geometry has emerged as a major tool for modeling the spatial distribution of wireless networks. In 5G and beyond networks, the role of stochastic geometry is envisioned to be more prominent due to the increased heterogeneity and complexity of future wireless networks \cite{hmamouche2021new}.
	
	A fundamental measure for the delay in a large-scale wireless network is the \textit{local delay}, which is defined as the average time, measured in the number of time slots, that is required for a successful packet transmission between a typical transmitter-receiver pair.
	With the aid of stochastic geometry, the authors of \cite{baccelli2010new,haenggi2012local,gong2013local} characterize the local delay in a mobile ad-hoc network under different transmitter-receiver association criteria, such as bipolar, nearest-neighbor, nearest-transmitter, and nearest-receiver models. 
	The local delay in heterogeneous cellular networks is derived in \cite{nie2015hetnets, zhang2015delay}. The authors of \cite{nie2015hetnets}, in particular, compute the local delay and energy efficiency under a finite local delay constraint in a downlink heterogeneous cellular network with a discontinuous transmission (DTX) scheme, where a BS is randomly turned off in a given time slot to reduce the interference and energy consumption. 
	In \cite{zhang2015delay}, on the other hand, the authors study the trade-off between the delay and reliability by taking into account the random BS activity. 
	It is also worth mentioning that the local delay can be characterized as the $-1$-th moment generating function of the SIR/SINR meta  distribution. For Poisson cellular networks and SIR-based coverage, it is given by the inverse of a hypergeometric function \cite{haenggi2015meta}. The meta distribution of the SINR-based coverage criterion in cellular networks is derived in  \cite{yang2019meta, salehi2017analysis, deng2018meta}.
	
	Despite the relative ease of analysis and  tractability, the local delay has some inherent limitations. First of all, under some circumstances, a network may exhibit a phenomena called wireless contention phase transition \cite{baccelli2010new}, according to which the local delay is infinite due to the contributions from a few users located in ``bad spots''.
	Therefore, the evaluation of the performance based on the local delay alone renders networks with an infinite delay unreliable, despite the fact that these networks may provide good service in terms of delay for most of their users. In 5G and beyond networks, furthermore, mission-critical applications such as vehicular communications and remote surgery are envisioned to be among the most important use cases in which, rather than the average delay, more emphasis is put on the delay distributions and percentiles. Unfortunately, these fine-grained information cannot be provided by the traditional definition of local delay, since it provides information only about the average delay.

	\subsection{Contributions}
	
	Motivated by the limitations of the local delay and the need for more relevant metrics for characterizing the delay for the applications envisioned in 5G and beyond communications, we focus our attention of the distribution of the delay rather than on the average delay. More precisely, we introduce and compute analytical frameworks for three definitions of delay distribution. Each definition of delay distribution uniquely captures fine-grained delay information for the users. To overcome some computational challenges of the obtained analytical formulations, we also propose efficient numerical approximations for computing the delay distributions in practice. To the best of our knowledge, this work offers the first comprehensive analytical formulation of the delay distribution in large-scale cellular networks by using stochastic geometry. 
	
	In summary, the main contributions of our paper are as follows:
	\begin{itemize}
		\item We introduce a framework for analyzing a cellular network modeled as a homogeneous PPP where the active-inactive status of each BS depends on the existence of a user in the cell of the BS. We introduce different coverage criteria and compute the local delay for each criterion. In this paper, we consider three coverage criteria:
		\begin{enumerate}
			\item \textit{SIR-based coverage}: Under this criterion, a transmission is successful if the SIR at the user is greater than a given decoding threshold.
			\item \textit{SINR-based coverage}: Under this criterion, a transmission is successful if the SINR at the user is greater than a given decoding threshold.
			\item \textit{SIR+ASNR-based coverage}: Under this criterion, a transmission is successful if the SIR at the user is greater than a given decoding threshold \textit{and} the ASNR is greater than a given detection threshold.
		\end{enumerate}		
		
		\item We consider three different delay distribution metrics and derive analytical expressions for each distribution and for the SIR-based, SINR-based, and SINR+ASNR-based definitions of coverage. In particular, the following delay distribution metrics are analyzed:
		\begin{enumerate}
			\item The first distribution is formulated as the expectation, over all network spatial realizations, of the complementary cumulative distribution function (CCDF) of the required number of time slots for successful transmission over all possible channel realizations.
			
			\item The second distribution is formulated as the CCDF of the expectation of the required number of time slots for successful transmission conditioned on a fixed network realization.
			
			\item The third distribution is formulated as the CCDF, of the CCDF conditioned on the network realization, of the required number of time slots for successful transmission.
		\end{enumerate}
		
		\item As for the first distribution, we analyze the asymptotic packet loss probability. We focus our attention on the event that a typical user never receives its packet regardless of the waiting time. We show that the packet loss probability is zero when the SIR-based or the SINR-based coverage is considered, and it is non-zero when the SIR+ASNR-based coverage is considered.
		
		\item To overcome some numerical difficulties that arise in the computation of the obtained delay distribution, we introduce efficient numerical approximations. In particular, we show that the second and the third delay distributions can be approximated by using the numerical inversion of the Laplace transform and (when the SIR-based or the SINR-based coverage are considered) the Beta distribution. We also show that the first delay distribution can be expressed in terms of an integration of the third distribution, and thus it can be efficiently approximated by using the Riemann sum.
	\end{itemize}

	\subsection{Paper Organization}
	
	The rest of this paper is organized as follows. In Section \ref{sec:System-Model}, we introduce the system model.
	In Section \ref{sec:Preliminaries}, we give preliminary results and definitions for the three delay distributions.	
	In Section \ref{sec:Local-Delay}, we provide the analysis of the local delay under different coverage criteria and show that, under some conditions, it is finite under the SIR-based coverage, while it is infinite for the other coverage criteria.
	In Sections \ref{sec:F1-distribution} and \ref{sec:F2-F3-distribution}, we introduce analytical expressions of the delay distributions for different coverage criteria. In particular, in Section \ref{sec:F1-distribution}, we provide the asymptotic analysis of the first distribution which corresponds to the packet loss probability.
	In Section \ref{sec:Approximations}, we propose efficient numerical approximations for the delay distributions based on techniques such as the Euler-sum method, the Beta distribution, and the Riemann sum, so as to overcome some inherent numerical difficulties in the evaluation of the delay distributions. Finally, we provide numerical results in Section \ref{sec:Numerical-Results} to validate our findings, and we conclude the paper in Section \ref{sec:Conclusion}.
	The main notation used in this paper is listed in Table \ref{table:notation}.
	
	\begin{table*}[t] 
		\footnotesize
		\centering
		\caption{Main symbols and functions used in the paper  \vspace{-0.25cm}}
		\newcommand{\tabincell}[2]{\begin{tabular}{@{}#1@{}}#2\end{tabular}}
		\begin{tabular}{|l||l|} \hline
			Symbol/Function & Definition \\ \hline \hline
			$\PhiMT$, $\PhiBS$, $\PhiI$, $\Phi$ & PPP of MTs, BSs, interfering BSs, the whole network\\ \hline
			$\lambdaMT$, $\lambdaBS$, $\lambda$ & Density of mobile terminals, base stations, the whole network\\ \hline
			$\mt_0$, $\bs_0$ & Typical user, base station serving the typical user \\ \hline
			$h$, $\ell(r)$, $K > 0$, $\alpha > 2$ & Fading channel, path-loss function, path-loss constant, path-loss exponent \\ \hline
			$f(r_{0}) = 2\pi{}\lambdaBS r_{0}\exp({-\lambdaBS\pi{}r_{0}^{2}})$ & Probability density function of the nearest BS distance \\ \hline
			$\e[\cdot]$, $\pr[\cdot]$ & Expectation operator, probability measure \\ \hline
			$\e_\Phi[\cdot]$, $\e_h[\cdot]$ & Expectation over network realization, expectation over fading realization \\ \hline
			$\pr_\Phi[\cdot]$, $\pr_h[\cdot]$ & Probability over network realization, probability over fading \\ \hline
			$L(\lambdaBS,\lambdaMT)= 1- \left(1+{\lambdaMT}/({3.5}{\lambdaBS})\right)^{-3.5}$ & Active probability of a base station\\ \hline
			$P$, $S$, $I$, $W$ & Transmit power, signal, interference, noise \\ \hline
			$\gamma$, $\theta$ & decoding threshold, detection threshold \\ \hline
			$\Delta_{\Phi}$, $\pcovcon$ & Conditional number of time slots, conditional coverage probability \\ \hline
			$C^\tau_k = \tau!/(k!(\tau-k)!)$ & Binomial coefficient \\\hline
			${}_2F_1(\cdot,\cdot;\cdot;\cdot)$, $\Gamma(\cdot)$, $\mathbbm{1}(\cdot)$ & Hypergeometric function, Gamma function, indicator function \\ \hline
			$\varphi_Z(\cdot)$, $\mathcal{M}_Y(\cdot)$ & Characteristic function of $Z$, moment-generating function of $Y$ \\ \hline
			$\Im\{\cdot\}$, $\Re\{\cdot\}$ & Operators returning the imaginary part, real part of the argument \\ \hline
			$D$, $F_1(\tau)$, $F_2(T)$, $F_3(x,\tau)$ & Local delay, $F_1$, $F_2$, $F_3$ delay distribution \\ \hline
			$\mathcal{F}\left(k,\alpha,\gamma\right)$ & Shorthand for $ 1 +  L(\lambdaBS,\lambdaMT)\left({}_2F_1\left(-{2}/{\alpha},k;1-{2}/{\alpha};-\gamma\right)-1\right) \geq 1$\\ \hline
		\end{tabular}
		\label{table:notation} \vspace{-0.5cm}
	\end{table*}

	\section{System Model}\label{sec:System-Model} 
	
	\subsection{Cellular Network Modeling} \label{subsec:CellularNetworkModeling} 
	
	We consider a downlink cellular network where the BSs are modeled as points of a homogeneous PPP, denoted by $\Phi_{\rm{BS}}$, with density $\lambda_{\rm{BS}}$ while the mobile terminals (MTs) are modeled as another homogeneous PPP, denoted by $\Phi_{\rm{MT}}$, with density $\lambda_{\rm{MT}}$. We assume that $\Phi_{\rm{BS}}$ and $\Phi_{\rm{MT}}$ are independent of each other. Thus, the whole network can be modeled as a point process $\Phi$ which is the joint set of BSs and MTs, i.e., $\Phi = \PhiBS \cup \PhiMT$. Due to the properties of homogeneous PPPs, $\Phi$ forms another homogeneous PPP with the intensity $\lambda = \lambdaBS + \lambdaMT$.
	Each BS transmits with constant power denoted by $P$. Leveraging Slivnyak's theorem \cite[Th. 1.4.5]{BaccelliBook2009}, we focus our analysis on the performance as seen by a typical user located at the origin, denoted by $\mt_0$, whose serving BS is denoted by ${\rm{BS}}_{0}$. 
	Throughout the paper, the subscripts ${0}$, ${i}$ and ${n}$ identify the intended link, a generic interfering link, and a generic BS-to-MT link. 
	It is worth mentioning that, although the network is modeled as a random variable, it is static. This means that, for a given network realization, the locations of the MTs and BSs do not change over time. 
	
	\subsection{Transmission Scheme} 
	
	We consider that time is divided into discrete slots with equal duration and each transmission attempt occupies one time slot. At the beginning of a time slot, a typical user, $\mt_0$, requests a packet to its serving BS, $\bs_0$. Depending on the medium access control (MAC) policy, the $\bs_0$ may or may not handle the request in the same time slot where the request is made. In this paper, we assume that $\bs_0$ will attempt a transmission of the packet requested by $\mt_0$ immediately. If the attempt fails due to an insufficient received signal quality, a re-transmission is scheduled in the next time slot until the user gets the requested packet.
	
	\subsection{Channel Modeling} \label{subsec:ChannelModeling} 
	
	For each BS-to-MT link, we model the wireless channel with a distance-dependent path-loss and fast-fading that takes into account the small-scale fading along the transmission path. Additionally, we consider that all BS-to-MT links are mutually independent and identically distributed (i.i.d.). We assume that the fading is quasi-static, that is, the fading coefficient is constant over a period of a time slot and it varies from one time slot to another according to the i.i.d. Rayleigh fading model. The power fading coefficient between $\mt_0$ and its serving BS, $\bs_0$ is denoted by $h_0$, while the coefficient between $\mt_0$ and an interfering base station, $\bs_i \in \PhiBS \setminus \bs_0$, is denoted by $h_i$. 
	The path-loss from a BS and an MT is modeled as
	\begin{align}
		\ell(r) = Kr^{-\alpha}
	\end{align}
	where $K > 0$ is the path-loss constant, $r > 0$ is the distance between the BS and the MT, and $\alpha > 2$ is the path-loss exponent.

	\subsection{Cell Association Criterion} \label{subsec:CellAssociationCriterion} 
	
	The cell association criterion is based on the highest average received power without considering the impact of fading. Thus, the association depends only on the deterministic path loss within a fixed network. Let $\bs_n \in \PhiBS$ denote a generic BS of the network. Thus, given a typical user, $\mt_0$, its serving BS, $\bs_0$, is identified as follows:
	\begin{equation}
		\bs_0 \mydef \arg \max_{\bs_n \in \PhiBS} \left\{\frac{1}{\ell(r_n)}\right\}.
	\end{equation}
	Let us denote the distance from $\mt_0$ to $\bs_0$ by $r_0$. Due to the randomness of the network, $r_0$ is a random variable whose probability density function (PDF) is given as follows:
	\begin{equation}\label{eq:PDF-of-r0}
		f(r_{0}) = \pr_{\Phi}[r = r_0]= 2\pi{}\lambdaBS r_{0}e^{-\lambdaBS\pi{}r_{0}^{2}}.
	\end{equation}
	where $\pr_{\Phi}[\cdot]$ indicates that the only random variable in evaluating the probability is the network geometry $\Phi$.
	In non-fully loaded scenarios, there exist some BSs that are not associated with any MT in the network. In such cases, the BSs are said to be \textit{inactive}. On the other hand, if a BS is associated with at least one MT in its cell, it is said to be \textit{active}. Given $\lambdaBS$ and $\lambdaMT$, the probability of a randomly selected BS being active, denoted by $L(\lambdaBS,\lambdaMT)$, is \cite{yu2013downlink}:
	\begin{equation}
		L(\lambdaBS,\lambdaMT) \approx 1- \left(1+\frac{1}{3.5}\frac{\lambdaMT}{\lambdaBS}\right)^{-3.5}
	\end{equation}
	Due to the static assumption of the network, for a given network realization $\Phi$, the active/inactive statuses of all BSs remain the same all the time. Therefore, the time-dependence of the interference is a function of only the power fading coefficients $h_{i}$.
	Let $\PhiI$ be the set of all interfering BSs, i.e. all active BSs other than $\bs_0$. Since $\PhiI$ is constructed from a homogeneous PPP with intensity $\lambdaBS$, the active probability of a generic BS is $L(\lambdaBS,\lambdaMT)$, and a single point at the origin is removed, it forms an inhomogeneous point process \cite{di2016intensity} with density $\lambda^*(r)$ defined as:
	\begin{equation}\label{eq:intensity-inhomogenous}
		\lambda^*(r) \approx
		\begin{cases}
			\lambdaBS L(\lambdaBS,\lambdaMT) & r > r_0\\
			0 & r \leq r_0.
		\end{cases}
	\end{equation}

	\section{Preliminaries}
	\label{sec:Preliminaries} 
	
	In this section, we introduce necessary concepts for further analysis, such as the conditional coverage probability (CCP), the local delay, as well as the definitions of three delay distributions.

	\subsection{Conditional Coverage Probability}\label{subsec:pcovcon} 
	
	Consider the typical user, $\mt_0$, located at the origin. The signal from the serving BS, $\bs_0$, and the aggregated signals from the interfering BSs, $\cup \bs_i  = \Phi_I$, are denoted by $S$ and $I$, respectively, and they are given as follows:
	\begin{equation}
		S = \frac{Ph_{0}}{\ell(r_0)}, \qquad I = \sum_{i \in \PhiI} \frac{Ph_{i}}{\ell(r_i)}.
	\end{equation}
	Using these notations, the SIR, SINR, and ASNR at $\mt_0$ are given, respectively, as follows:
	\begin{equation} \label{eq:sir-sinr-asnr}
		\sir = \frac{S}{I}, \qquad 
		\sinr = \frac{S}{I+W}, \quad
		\asnr = \frac{P}{W\ell(r_0)}
	\end{equation}
	where $W > 0$ is a deterministic real number which represents the variance of the Gaussian noise power density. Note that $\sir$, $\sinr$, and $\asnr$ are random variables. In particular, the randomness of $\sir$ and $\sinr$ comes from the network geometry $\Phi$ and the fading coefficients $h_0$ and $h_i$, while the randomness of $\asnr$ comes only from $\Phi$.	
	
	In general, the CCP is defined as the probability of the event that $\mt_0$ is successfully served by $\bs_0$ in a given time slot, with the (spatial) realization of $\Phi$ being fixed. The specific expression of the CCP depends on the criterion used to define a successful transmission (i.e., the coverage probability). Throughout this paper, we consider three different coverage criteria: (i) SIR-based, (ii) SINR-based, and (iii) SIR+ASNR-based. 
	To this end, let $\pcovcon$ denote the CCP for a given realization of $\Phi$ and let $\pr_h[\cdot]$ and $\e_h[\cdot]$ be the probability and expectation operator that treat the fading coefficients $h_0$ and $h_i$ as the only random variables (while the other variables are fixed). 
	Using these notations, the definition of the CCP for each criterion is given as follows.
	\subsubsection{SIR-based coverage}
	Under this criterion, the $\mt_0$ is said to be covered by $\bs_0$ if, given $\Phi$, the SIR is greater than a given decoding threshold $\gamma$. Thus, the CCP is formulated as:
	\begin{equation}\label{eq:covcon-SIR}
		\pcovcon = \pr_h\left[\sir > \gamma~|~\Phi\right].
	\end{equation}
	
	\subsubsection{SINR-based coverage}
	Under this criterion, the $\mt_0$ is said to be covered by $\bs_0$ if, given $\Phi$, the SNIR is greater than a given decoding threshold $\gamma$. Thus, the CCP is formulated as:
	\begin{equation}\label{eq:covcon-SINR}
		\pcovcon = \pr_h\left[\sinr > \gamma~|~\Phi\right].
	\end{equation}
	If $W=0$ in \eqref{eq:sir-sinr-asnr}, the SINR is equal to the SIR, and thus \eqref{eq:covcon-SINR} simplifies to \eqref{eq:covcon-SIR}.
	
	\subsubsection{SIR+ASNR-based coverage}
	Under this criterion, the $\mt_0$ is said to be covered by $\bs_0$ if, given $\Phi$, the SIR is greater than a given decoding threshold $\gamma$ and the ASNR is greater than a given detection threshold $\theta$. Thus, the CCP is formulated as:
	\begin{equation}\label{eq:covcon-SIR+ASNR}
		\pcovcon = \pr_h\left[\sir > \gamma, \asnr > \theta~|~\Phi\right].
	\end{equation}
	If $W=0$ in \eqref{eq:sir-sinr-asnr}, the ASNR becomes infinite. Thus, the second condition in \eqref{eq:covcon-SIR+ASNR} always holds and the definition of the CCP coincides with \eqref{eq:covcon-SIR}.

	\subsection{Local Delay}
	
	Given $\Phi$, let $\Delta_{\Phi}$ be the number of time slots required for a successful transmission, i.e., until $\mt_0$ is covered by $\bs_0$ and can successfully decode the data packet intended to it.
	Due to the time-independence of the channels between $\mt_0$ and all the BSs in the network, the success probability for any time slot is equal to $\pcovcon$.
	In this case, the packet transmission from $\bs_0$ to $\mt_0$ can be regarded as a Bernoulli trial with success probability $\pcovcon$.
	Therefore, $\Delta_{\Phi}$ is a geometrically distributed random variable whose probability mass function (PMF) and expected value are given as follows:
	\begin{equation}\label{eq:pdf-mean-time-slots}
		\pr_h\left[\Delta_{\Phi} = k\right] = \pcovcon (1-\pcovcon)^k, \qquad \e_h\left[\Delta_{\Phi}\right] = (\pcovcon)^{-1}
	\end{equation}
	The \textit{local delay}, denoted by $D$, is defined as the expected number of time slots required for successful transmission. Thus, it can be obtained by spatially averaging $\Delta_{\Phi}$, i.e.,
	\begin{equation}\label{eq:local-delay-def}
		D = \e_{\Phi}\left[\e_h\left[\Delta_{\Phi}\right]\right] = \e_{\Phi}\left[\frac{1}{{\pcovcon}}\right].
	\end{equation}	
	
	\subsection{Delay Distribution}
	
	As discussed in Section \ref{sec:Introduction}, the local delay may have some limitations for characterizing the delay reliability of cellular networks, e.g., it provides information only on the average delay. To overcome these limitations, we propose three delay distribution metrics, where each metric captures some unique statistical information. The mathematical definitions of the delay distributions and their relations to relevant performance metrics in wireless communications are elaborated as follows.
	
	\subsubsection{$F_1$ distribution}
	
	Given $\tau \in \mathbb{Z}^+$, the $F_1$ distribution is defined as follows:
	\begin{equation}\label{eq:Delay-distribution-F1}
		F_1(\tau) \mydef \e_{\Phi}\left[\pr_{h} \left[\Delta_\Phi > \tau~|~\Phi\right]\right].
	\end{equation}	
	In the context of wireless communications, $F_1(\tau)$ is the probability that a typical user experiences a delay greater than the predetermined time deadline $\tau$.
	Therefore, it can be regarded as some kind of delay violation probability. 
	Equivalently, it can also be regarded as the expected fraction of users in the network that experience a delay greater than $\tau$ time slots. 
	It is worth mentioning that the local delay can be expressed in terms of the $F_1$ distribution.
	\begin{corollary}\label{corollary:delay-F1-relation}
		The local delay and the $F_1$ distribution are related as follows:
		\begin{equation}\label{eq:delay-F1-relation}
			D = \sum_{\tau=0}^{\infty} F_1(\tau).
		\end{equation}
	\end{corollary}
	\begin{proof}
		See Appendix \ref{appendix:proof-of-delay-F1-relation}.
	\end{proof}

	\subsubsection{$F_2$ distribution}
	
	Given $T \in \mathbb{R}^+$, the $F_2$ distribution is defined as follows:
	\begin{align}\label{eq:Delay-distribution-F2}
		F_2(T) 
		&\mydef \pr_\Phi\left[\e_h\left[\Delta_\Phi~|~\Phi\right] \geq T\right] 
	\end{align}
	In the context of wireless communications, $F_2(T)$ is the probability that the expected delay, conditioned on a network realization and averaged over the fading channels, is greater than the threshold $T$. Similar to the $F_1$ distribution, it can be seen as a kind of delay violation probability.

	\subsubsection{$F_3$ distribution}
	Given $\tau \in \mathbb{Z}^+$ and $x \in [0,1]$, the $F_3$ distribution is defined as follows:
	\begin{equation}\label{eq:Delay-distribution-F3}
		F_3(x,\tau) = \pr_\Phi\left[\pr_h\left[\Delta_\Phi > \tau~|~\Phi\right] \geq x\right].
	\end{equation}
	To understand the physical interpretation of $F_3(x,\tau)$ in wireless communications, consider a \textit{reliability} metric $R(p,\tau)$ which is a measure of the system capability to transmit a packet within a deadline $\tau$ and with success probability greater than $p$ \cite{bennis2018ultrareliable}. $F_3(x,\tau)$ is closely related to $R(p,\tau)$. 
	In particular, we have $R(p,\tau) 	= \pr_\Phi\left[\pr_h\left[\Delta_\Phi \leq \tau~|~\Phi\right] > p\right]	= 1 - F_3(1-p,\tau)$.
	
	\section{Analytical Formulation of the Local Delay}\label{sec:Local-Delay}
	
	In this section, we derive analytical expressions for the local delay for each coverage definition. 
	Then, we determine the condition under which the local delay is infinite. First, we present the following result on the formulation of the CCP.

	\begin{lemma}\label{lemma:coverage-probability-general-formula}
		Let $\mathcal{G}(r_0)$ be a coverage criterion-dependent function defined as:
		\begin{align}\label{eq:G-aux}
			\mathcal{G}(r_0) = 
			\begin{cases}
				1, & \textup{SIR-based coverage},\\
				e^{-\frac{\gamma WKr_0^\alpha}{P}}, & \textup{SINR-based coverage},\\
				\mathbbm{1}\left({r_0 \leq \left(\frac{P}{KW\theta}\right)^{1/\alpha}}\right), &  \textup{SIR+ASNR-based coverage}.
			\end{cases}
		\end{align}
		For any coverage criterion with $\mathcal{G}(r_0)$ defined in \eqref{eq:G-aux}, the CCP is formulated as follows:
		\begin{equation}\label{eq:coverage-probability-general-formula}
			\pcovcon
			= \mathcal{G}(r_0)\prod_{i\in \PhiI} \left(1+\gamma\left(\frac{r_0}{r_i}\right)^\alpha\right)^{-1}.
		\end{equation}
	\end{lemma}
	\begin{proof}
		See Appendix \ref{appendix:proof-of-coverage-probability-general-formula}
	\end{proof}
	
	Using Lemma \ref{lemma:coverage-probability-general-formula}, we can formulate the local delay as follows. 
	
	\begin{theorem}\label{thm:local-delay-general}
		For any coverage criterion, the local delay is given by: 
		\begin{align}\label{eq:local-delay-general}
			D 
			&= 2\pi\lambdaBS\int_{0}^{\infty} \frac{r_0 e^{-\pi\lambdaBS r_0^2\left(1-\frac{2\gamma L(\lambdaBS,\lambdaMT)}{\alpha-2}\right)}}{\mathcal{G}\left(r_0\right)}  dr_0.
		\end{align}
	\end{theorem}
	\begin{proof}
		See Appendix \ref{appendix:local-delay-general}.
	\end{proof}
	
	The local delay in \eqref{eq:local-delay-general} is given as a simple integral. If we consider the SIR-based coverage, we can obtain a closed-form expression under some conditions.
	\begin{corollary}\label{corollary:local-delay-SIR}
		Assume $1-2\gamma L(\lambdaBS,\lambdaMT)/(\alpha-2) >0$. Under the SIR-based coverage criterion, the local delay has the following closed-form expression:
		\begin{equation}
			D = \left(1-\frac{2\gamma L(\lambdaBS,\lambdaMT)}{\alpha-2}\right)^{-1}.
		\end{equation}
		If $1-2\gamma L(\lambdaBS,\lambdaMT)/(\alpha-2) \leq 0$, on the other hand, $D$ is infinite.
	\end{corollary}
	\begin{proof}
		According to \eqref{eq:G-aux}, under the SIR-based criterion, we have $\mathcal{G}(r_0) = 1$. Thus, \eqref{eq:local-delay-general} becomes:
		\begin{align}\label{eq:local-delay-SIR}
			D = 2\pi\lambdaBS \int_{0}^{\infty} r_0e^{-\pi\lambdaBS r_0^2\left(1-\frac{2\gamma L(\lambdaBS,\lambdaMT)}{\alpha-2}\right)} dr_0. 
		\end{align}
		The integral is finite when $1-\frac{2\gamma L(\lambdaBS,\lambdaMT) }{\alpha-2}>0$. Otherwise, it is infinite since the exponential term of the integrand function grows without bound. Computing the integral completes the proof.
	\end{proof}
	
	\begin{corollary}\label{corollary:local-delay-SINR-and-SIR+ASNR}
		Under the SINR-based or the SIR+ASNR-based criterion, the local delay is infinite.
	\end{corollary}
	
	\begin{proof}
		According to \eqref{eq:G-aux}, under the SINR-based criterion, we have $\mathcal{G}(r_0) = e^{-\frac{\gamma WKr_0^\alpha}{P}}$. Thus, the local delay in \eqref{eq:local-delay-general} reduces to:
		\begin{equation}
			D = 2\pi\lambdaBS \int_{0}^{\infty} r_0  e^{\frac{\gamma WKr_0^\alpha}{P}-\pi\lambdaBS r_0^2\left(1-\frac{2\gamma L(\lambdaBS,\lambdaMT)}{\alpha-2}\right)} dr_0 \label{eq:local-delay-SINR}
		\end{equation}
		The integral expression is infinite due to the assumption $\alpha > 2$. When the SIR+ASNR-based coverage is considered, we have $\mathcal{G} = \mathbbm{1}\left({r_0 \leq \left(\frac{P}{KW\theta}\right)^{1/\alpha}}\right)$ and the local delay in \eqref{eq:local-delay-general} becomes:
		\begin{align} \label{eq:local-delay-SIR+ASNR}
			D &
			= 2\pi\lambdaBS \int_0^{\infty} \frac{ r_0e^{-\pi\lambdaBS r_0^2 \left(1-\frac{2\gamma L(\lambdaBS,\lambdaMT)}{\alpha-2}\right)}}{\mathbbm{1}\left({r_0 \leq \left(\frac{P}{KW\theta}\right)^{1/\alpha}}\right)}dr_0 \\
			=& 2\pi\lambdaBS \left[\int_0^{\left(\frac{P}{KW\theta}\right)^{\frac{1}{\alpha}}} \frac{ r_0e^{-\pi\lambdaBS r_0^2 \left(1-\frac{2\gamma L(\lambdaBS,\lambdaMT)}{\alpha-2}\right)}}{\mathbbm{1}\left({r_0 \leq \left(\frac{P}{KW\theta}\right)^{1/\alpha}}\right)}dr_0 + \int_{\left(\frac{P}{KW\theta}\right)^{\frac{1}{\alpha}}}^{\infty} \frac{r_0e^{-\pi\lambdaBS r_0^2 \left(1-\frac{2\gamma L(\lambdaBS,\lambdaMT)}{\alpha-2}\right)}}{\mathbbm{1}\left({r_0 \leq \left(\frac{P}{KW\theta}\right)^{1/\alpha}}\right)}dr_0 \right]. \nonumber
		\end{align}
		The last integral on the right hand side of \eqref{eq:local-delay-SIR+ASNR} is infinite since $\mathbbm{1}\left({r_0 \leq \left(\frac{P}{KW\theta}\right)^{1/\alpha}}\right)$ becomes zero and hence $D$ is infinite. This completes the proof.
	\end{proof}	
	
	\begin{remark}
		Let $\gamma^*$ be the SIR critical threshold defined as
		\begin{align}\label{eq:critical-treshold}
			\gamma^* = \frac{\alpha-2}{2L(\lambdaBS,\lambdaMT)}
		\end{align}
		From Corollary \ref{corollary:local-delay-SIR}, under the SIR-based criterion, the local delay is finite and positive if $\sir < \gamma^*$, and infinite otherwise. On the other hand, according to Corollary \ref{corollary:local-delay-SINR-and-SIR+ASNR}, the local delay is always infinite under the SINR-based or the SIR+ASNR-based criterion.
	\end{remark}

	\section{Analysis of the $F_1$ Distribution}
	\label{sec:F1-distribution}
	
	In this section, we derive the general formulation of the $F_1$ distribution, as well as its analytical expression for each considered coverage criterion. We also provide the asymptotic analysis of the distribution which corresponds to the packet loss probability.
	
	\subsection{Analytical Formulation of the $F_1$ Distribution}
	
	The general formulation of the $F_1$ distribution is given in the following theorem.
	\begin{theorem}\label{thm:F1-distribution-general}
		For any coverage criterion, the $F_1$ distribution is formulated as:
		\begin{align}\label{eq:F1-distribution-general}
			F_1(\tau) 
			&=  2\pi\lambdaBS \sum_{k=0}^{\tau} C_{k}^{\tau}(-1)^{k} \int_0^\infty \left(\mathcal{G}(r_0)\right)^k  e^{-\pi\lambdaBS r_0^2 \mathcal{F}\left(k,\alpha,\gamma\right) }  r_0dr_0.
		\end{align}
	\end{theorem}
	
	\begin{proof}
		See Appendix \ref{appendix:proof-of-F1-distribution-general}.
	\end{proof}
	
	Based on Theorem \ref{thm:F1-distribution-general}, the analytical expressions of the $F_1$ distribution for the three considered coverage criteria are given by the following corollary.
	
	\begin{corollary}\label{corollary:F1-distribution-exact}
		Let $F_1^{\textup{sir}}(\tau)$, $F_1^{\textup{sinr}}(\tau)$, and $F_1^{\textup{sir+asnr}}(\tau)$ denote the $F_1$ distribution under the assumption of the SIR, SINR, and SIR+ASNR-based coverage criterion, respectively.
		The analytical expression of $F_1(\tau)$ for each criterion can be computed as follows:
		\begin{align}
			F_1^{\textup{sir}}(\tau) 
			&=  \sum_{k=0}^{\tau} \frac{C_{k}^{\tau}(-1)^{k}}{\mathcal{F}\left(k,\alpha,\gamma\right)} \label{eq:F1-distribution-SIR}\\
			F_1^{\textup{sinr}}(\tau) 
			&= \sum_{k=0}^{\tau} C_{k}^{\tau}(-1)^{k} 2\pi\lambdaBS \int_0^\infty e^{-\pi\lambdaBS r_0^2 \mathcal{F}\left(k,\alpha,\gamma\right)} e^{-k\frac{\eta WKr_0^\alpha}{P}}r_0dr_0 \label{eq:F1-distribution-SINR}\\
			F_1^{\textup{sir+asnr}}(\tau) 
			&= 1 + \sum_{k=1}^{\tau} \frac{C_{k}^{\tau}(-1)^{k}}{\mathcal{F}\left(k,\alpha,\gamma\right)}  \left[{1-e^{-\pi\lambdaBS \left(\frac{P}{KW\theta}\right)^{\frac{2}{\alpha}} \mathcal{F}\left(k,\alpha,\gamma\right) }}\right]. \label{eq:F1-distribution-SIR+ASNR}
		\end{align}
	\end{corollary}
	\begin{proof}
		$F_1^{\textup{sir}}(\tau)$ and $F_1^{\textup{sinr}}(\tau)$ follow by inserting $\mathcal{G}(r_0)$ in \eqref{eq:G-aux} into \eqref{eq:F1-distribution-general}. 
		As for $F_1^{\textup{sir+asnr}}(\tau)$, the proof can be found in Appendix \ref{appendix:proof-of-F1-distribution-exact}.
	\end{proof}
	
	\begin{remark}
		From Corollary \ref{corollary:F1-distribution-exact}, we can see that, under the SIR-based or SIR+ASNR-based criterion, $F_1(\tau)$ is given in a closed-form expression. Under the SINR-based criterion, on the other hand, $F_1(\tau) $ is not formulated, in general, in a closed-form expression.
		Since $\mathcal{F}\left(k,\alpha,\gamma\right) \geq 1$ and \eqref{eq:F1-distribution-SIR}, \eqref{eq:F1-distribution-SINR}, and \eqref{eq:F1-distribution-SIR+ASNR} are expressed as a finite sum, the $F_1$ distribution is always finite for all coverage criteria, even though the corresponding local delay may be infinite.
	\end{remark}
	
	\subsection{Asymptotic Analysis of the $F_1$ Distribution}
	
	In this section, we analyze the asymptotic behavior of the $F_1$ distribution as $\tau \to \infty$. In other words, we are interested in computing the following limit: $P_e = \lim_{\tau\to\infty} F_1(\tau)$.
	Recall that $F_1(\tau)$ can be regarded as the delay violation probability within the deadline $\tau$. Thus, in wireless communications, $P_e$ corresponds to the \textit{packet loss probability}, i.e., the probability that a user never receives its packet within an infinite waiting time. We first introduce the following results.
	
	\begin{lemma}\label{lemma:positivity-of-pcovcon}
		Under the SIR-based coverage criterion, we have $\pcovcon> 0$.
	\end{lemma}
	
	\begin{proof}
		See Appendix \ref{appendix:positivity-of-pcovcon}.
	\end{proof}
	
	\begin{proposition}\label{proposition:limit-of-F1-general-case}
		For any coverage criterion, the packet loss probability is:
		\begin{equation}\label{eq:limit-of-F1-general-case}
			P_e = \pr\left[\pcovcon = 0\right] = \pr[\mathcal{G}(r_0) = 0].
		\end{equation}
	\end{proposition}
	\begin{proof}
		See Appendix \ref{appendix:limit-of-F1-general-case}.
	\end{proof}
	
	Proposition \ref{proposition:limit-of-F1-general-case} implies that the expected packet loss probability is equivalent to the probability of the CCP being zero, which is equivalent to the probability that $\mathcal{G}(r_0)$ is zero.
	
	\begin{corollary}\label{corollary:limit-of-F1-SIR-and-SINR}
		Under the SIR-based or SINR-based coverage criterion, we have $P_e = 0$.
	\end{corollary}
	
	\begin{proof}
		According to \eqref{eq:G-aux}, under the SIR-based coverage criterion, $\mathcal{G}(r_0) = 1$. In this case, according to \eqref{eq:limit-of-F1-general-case}, the packet loss probability is $P_e = \pr[\mathcal{G}(r_0) = 0] = 0$. This is in agreement with Lemma \ref{lemma:positivity-of-pcovcon} as well. Under the SINR-based criterion, we have $\mathcal{G}(r_0) = e^{-\frac{\gamma WKr_0^\alpha}{P}}$. Thus:
		\begin{align}
			P_e 
			&= \pr\left[e^{-\frac{\gamma WK{r_0}^\alpha}{P}}=0\right] = \lim_{r\to\infty}\pr[r_0=r] \mystepA 0.
		\end{align}
		where $(a)$ follows from \eqref{eq:PDF-of-r0} and the fact that $\lim_{r\to\infty} r\exp\left(-\pi\lambdaBS r^2\right) = 0$.
	\end{proof}
	
	\begin{corollary}\label{corollary:limit-of-F1-SIR+ASNR}
		Under the SIR+ASNR-based coverage criterion, we have:
		\begin{align}
			P_e  &= e^{-\pi\lambdaBS\left(\frac{P}{KW\theta}\right)^{2/\alpha}}.
		\end{align}
	\end{corollary}
	\begin{proof}
		According to \eqref{eq:G-aux}, under the SIR+ASNR criterion, we have $\mathcal{G}(r_0) = \mathbbm{1}\left(r_0 \leq \left(\frac{P}{KW\theta}\right)^{1/\alpha}\right)$. Thus, from \eqref{eq:limit-of-F1-general-case}, we have $P_e = \pr\left[\mathbbm{1}\left(r_0 \leq \left(\frac{P}{KW\theta}\right)^{1/\alpha}\right)=0\right] = \pr\left[r_0 > \left(\frac{P}{KW\theta}\right)^{1/\alpha}\right]$.
		The proof follows by computing the CCDF of $r_0$ (see, e.g., \cite{yu2013downlink}). 
	\end{proof}
	
	\begin{remark}
		From Corollaries \ref{corollary:limit-of-F1-SIR-and-SINR} and \ref{corollary:limit-of-F1-SIR+ASNR}, we can evince the following:
		\begin{itemize}
			\item Under the SIR-based or SINR-based criterion, as $\tau \to \infty$, all the users are guaranteed to be served eventually, even though the local delay in these cases can be infinite.
			
			\item On the other hand, since $ e^{-\pi\lambdaBS\left(\frac{P}{KW\theta}\right)^{2/\alpha}} > 0$ if $\theta > 0$, the packet loss probability under the SIR+ASNR-based criterion is always non-zero. This implies that there is a certain portion of the users who, on average, will not get their packet no matter how long they wait.	
		\end{itemize}
	\end{remark}
	
	\begin{figure}[!t]
		\begin{center}
			\includegraphics[width=0.75\linewidth]{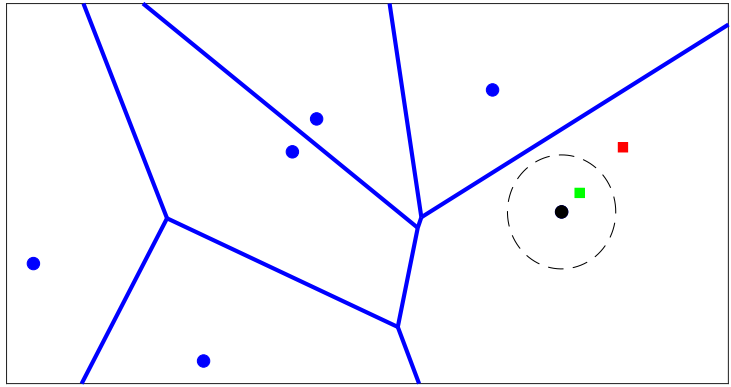}
			\caption{Illustration on the non-zero packet loss probability under the SIR+ASNR-based coverage criterion. Each dot and square represents a BS and an MT, respectively. Both the red and green MTs are associated to the black BS. The dashed circle represents the points with distance $r^* = \left(\frac{P}{KW\theta}\right)^{1/\alpha}$ from the green BS. The red MT will never receive its packet due to failure in fulfilling the coverage condition.}
			\label{fig:network-with-SNR}
		\end{center} \vspace{-1cm}
	\end{figure}	
	
	The non-zero probability of packet loss under the SIR+ASNR-based coverage criterion can be explained as follows. A necessary condition for a typical user $\mt_0$ to be served by its tagged BS, $\bs_0$, in a given time slot, is that the ASNR is greater than the detection threshold, i.e., $\frac{P}{WKr_0^{\alpha}} > \theta$. For given locations of $\mt_0$ and $\bs_0$, the quantity of $\frac{P}{WKr_0^{\alpha}}$ is fixed.  Therefore, given a pair $\bs_0$ and $\mt_0$, the $\mt_0$ will never receive its packet (i.e., it experiences a packet loss), if $r_0 > r^*$ where $r^* = \left(\frac{P}{KW\theta}\right)^{1/\alpha} $ (see Fig. \ref{fig:network-with-SNR} for illustration). In contrast, the condition for which the user is served according to the SINR-based coverage criterion is $\frac{P|h_0|/Kr_0^{\alpha}}{\sum_{i\in\PhiI} P|h_i|/Kr_i^{\alpha} + W} > \gamma$. Due to the random fading coefficients $|h_0|$ and $|h_i|$, there is a non-zero probability for this condition to be fulfilled, which results in a zero probability of packet loss for a long enough waiting time.

	\section{Analysis of the $F_2$ and $F_3$ Distributions}
	\label{sec:F2-F3-distribution}
	
	\subsection{Analytical Formulation of the $F_2$ and $F_3$ Distributions}
	
	To obtain analytical expressions for the $F_2$ and $F_3$ distributions, we first state the following preliminary result.
	
	\begin{lemma}\label{lemma:characteristic-function}
		Define the random variable $Z =\log(\pcovcon)$. For any coverage criterion, the characteristic function of $Z$, denoted by $\varphi_Z(t)$, is given as follows:
		\begin{equation}\label{eq:characteristic-function-general}
			\varphi_Z(t)  
			= 2\pi\lambdaBS\int_0^{\infty} \left(\mathcal{G}(r_0)\right)^{it}e^{-\pi\lambdaBS {r_0}^2 \mathcal{F}\left(it,\alpha,\gamma\right) } r_0 dr_0.
		\end{equation}
	\end{lemma}
	
	\begin{proof}
		By definition, the characteristic function of $Z =\log(\pcovcon)$ is:
		\begin{align}\label{eq:characteristic-function-definition}
			\varphi_Z(t) \mydef \e_{\Phi}\left[e^{itZ}\right] = \e_{\Phi}\left[\left(\pcovcon\right)^{it}\right] 
		\end{align}
		By substituting the conditional coverage probability $\pcovcon$ in \eqref{eq:coverage-probability-general-formula} into 
		\eqref{eq:characteristic-function-definition}, we obtain:
		\begin{align}
			\varphi_Z(t) 
			&= 
			\e_{\Phi}\left[\left(\mathcal{G}(r_0)\prod_{i\in \PhiI} \left(1+\gamma\left(\frac{r_0}{r_i}\right)^\alpha\right)^{-1}\right)^{it}\right]  \nonumber\\
			&\mystepA \e_{r_0}\left[\e_{\PhiI}\left[\prod_{i\in\PhiI} \left(1+\gamma\left(\frac{r_0}{r_i}\right)^\alpha\right)^{-it}\Bigr|~r_0\right]\left(\mathcal{G}(r_0)\right)^{it}\right].\nonumber
		\end{align}
		where $(a)$ follows from separating the conditioning on $\Phi$ into $\PhiI$ and $r_0$.		
		The proof follows by applying the probability generating functional (PGFL) theorem to the point process $\PhiI$ \cite{chiu2013stochastic} and by computing the expectation over $r_0$. 
	\end{proof}
	
	\begin{remark} \label{remark:characteristic-function-exact-expression}
		Let $\varphi_Z^{\textup{sir}}(t)$, $\varphi_Z^{\textup{sinr}}(t)$, and $\varphi_Z^{\textup{sir+asnr}}(t)$ denote the characteristic function of $Z$, i.e., $\varphi_Z(t)$, that correspond to the SIR, SINR, and SIR+ASNR based coverage criteria, respectively.
		By substituting  $\mathcal{G}(r_0)$ in \eqref{eq:G-aux} into \eqref{eq:characteristic-function-general} for each of the coverage criteria, we obtain:
		\begin{align}
			\varphi_Z^{\textup{sir}}(t) 
			&= \frac{1}{\mathcal{F}\left(it,\alpha,\gamma\right)} \label{eq:characteristic-function-SIR}\\
			\varphi_Z^{\textup{sinr}}(t) 
			&= 2\pi\lambdaBS\int_0^{\infty} r_0e^{-\pi\lambdaBS {r_0}^2\mathcal{F}\left(it,\alpha,\gamma\right)}  \left(e^{-\frac{\gamma WK{r_0}^\alpha}{P}}\right)^{it}dr_0 \label{eq:characteristic-function-SINR}\\
			\varphi_Z^{\textup{sir+asnr}}(t) 
			&= \frac{1-e^{-\pi\lambdaBS\left(\frac{P}{KW\theta}\right)^{\frac{2}{\alpha}}\mathcal{F}\left(it,\alpha,\gamma\right)}}{\mathcal{F}\left(it,\alpha,\gamma\right)} \label{eq:characteristic-function-SIR+ASNR}
		\end{align}
		From \eqref{eq:characteristic-function-SIR}, \eqref{eq:characteristic-function-SINR}, and \eqref{eq:characteristic-function-SIR+ASNR}, we can see that, under the SIR-based or SIR+ASNR-based criterion, $\varphi_Z(t)$ is given in a closed-form expression. Under the SINR-based criterion, on the other hand, $\varphi_Z(t)$ is, in general, not given in a closed-form expression.
	\end{remark}

	\begin{theorem}\label{thm:F2-F3-distribution-general}
		The $F_2$ and $F_3$ distributions can be formulated as follows:
		\begin{align}
			F_2(T)
			&= \frac{1}{2} - \frac{1}{\pi}\int_0^{\infty} \frac{\Im\left[T^{it}\varphi_Z(t)\right]}{t}dt \label{eq:F2-distribution-general}\\
			F_3(x,\tau)
			&= \frac{1}{2} - \frac{1}{\pi}\int_0^{\infty} \frac{\Im\left[{\left(1-x^{1/\tau}\right)}^{it}\varphi_Z(t)\right]}{t}dt. \label{eq:F3-distribution-general}
		\end{align}

	\end{theorem}
	
	\begin{proof} 
		Recall that $\Delta_{\Phi}$ is a geometrically distributed random variable whose mean is $\left(\pcovcon\right)^{-1}$.
		Let $F_{Z}(\cdot)$ denote the CDF of $Z$.
		From  \eqref{eq:Delay-distribution-F2}, thus, we have $F_2(T) 
		= \pr_\Phi\left[\pcovcon \leq \frac{1}{T}\right] = \pr\left[Z \leq -\log T\right] = F_{Z}(-\log T)$.
		From \eqref{eq:Delay-distribution-F3}, also, we have $F_3(x,\tau)
		= \pr_\Phi\left[\left(1-\pcovcon\right)^\tau \geq x\right] = \pr\left[Z \leq \log\left(1-x^{1/\tau}\right)\right]$ 
		$= F_{Z}\left(1-x^{1/\tau}\right) $.
		The proof follows from the inversion theorem \cite{gil1951note} and substituting $z=-\log T$ and $z=1-x^{1/\tau}$ for the $F_2$ and $F_3$ distributions, respectively.
	\end{proof}

	\section{Approximations for the Delay Distributions} 
	\label{sec:Approximations}
	
	The analytical formulations for the distributions of the delay derived in the previous section may not be easy to be computed in some cases. In particular, the $F_1$ distribution in Theorem \ref{thm:F1-distribution-general} is difficult to compute for large values of $\tau$ due to the binomial coefficients. The difficulty of computing the $F_2$ and $F_3$ distributions in Theorem \ref{thm:F2-F3-distribution-general} lies in numerically computing, with high precision, the integral expressions in \eqref{eq:F2-distribution-general} and \eqref{eq:F3-distribution-general}. Motivated by these considerations, in this section, we develop accurate and efficient numerical approximations for computing the $F_1$, $F_2$, and $F_3$ distributions. 
	
	\subsection{Approximation for the $F_2$ and $F_3$ Distributions via the Inverse Laplace Transform}
	
	To facilitate the numerical computation of the $F_2$ and $F_3$ distributions, we utilize the numerical inversion of the Laplace transform based on the Euler-sum method. This method is based on \cite{abate1995numerical} and it was used in \cite{ko2000outage} to compute the outage probability of diversity systems over generalized fading channels. More recently, it was also used in \cite{di2014stochastic} and \cite{wang2019meta}. Using this technique, the CDF of a random variable $Y$ can be approximated as
	\begin{align}\label{eq:inversion-of-Laplace-transform}
		F_{Y}(y)
		&= \frac{2^{-Q}e^{A/2}}{y}\sum_{q=0}^{Q}\sum_{n=0}^{N+q}C^Q_q\frac{(-1)^n}{\beta_n} \Re\left\{\hat{F}_Y\left(\frac{A+i2\pi n}{2y}\right)\right\} + E(A,Q,N).
	\end{align}
	where $\hat{F}_Y(s)$ is the Laplace transform of $F_Y(y)$, $E(A,Q,N)$ is the approximation error that is determined by the three parameters $A$, $N$, and $Q$, 
	and $\beta_n$ is defined as $\beta_n = 2$ if $n=0$ and $\beta_n = 1$ if $n > 0$. In this paper, we use the same parameters as in \cite{ko2000outage}, namely $A = 10\log 10 \approx 23.03$, which guarantees a discretization error of the order of $10^{-10}$, as well as $N = 21,$ and $Q=15$ to ensure that the resulting truncation error is less than $10^{-10}$.
	
	\begin{proposition}\label{proposition:F2-F3-approximation-inverse}
		The $F_2$ and $F_3$ distributions can be approximated as follows:
		\begin{align} 
			F_2(T)
			&\approx 1 - \frac{e^{A/2}}{2^{Q-1}}\sum_{q=0}^{Q}\sum_{n=0}^{N+q}C^Q_q\frac{(-1)^n}{\beta_n} \Re\left\{\frac{\varphi_Z\left(\frac{-iA+2\pi n}{2\log(T)}\right)}{(A+i2\pi n)}\right\}, \\
			F_3(x,\tau)
			&\approx 1 - \frac{e^{A/2}}{2^{Q-1}}\sum_{q=0}^{Q}\sum_{n=0}^{N+q}C^Q_q\frac{(-1)^n}{\beta_n}  \Re\left\{\frac{\varphi_Z\left(\frac{iA-2\pi n}{2\log(1-x^{1/\tau})}\right)}{(A+i2\pi n)}\right\}.
		\end{align}
	\end{proposition}
	
	\begin{proof}
		Define $Y = -\log(\pcovcon)$. From \eqref{eq:Delay-distribution-F2}, the $F_2$ distribution can be written as:
		\begin{eqnarray}\label{eq:F2-FY-relationship}
			F_2(T) 
			= \pr\left[Y \geq \log T\right]
			= 1 - F_{Y}(\log(T)).
		\end{eqnarray}
		Since $F_{Y}(0) = \pr[-\log(\pcovcon) < 0] = \pr\left[\pcovcon > 1\right] = 0$, the Laplace transform of $F_{Y}(y)$, i.e., $\hat{F}_{Y}(y)$, can be formulated as $\hat{F}_{Y}(s) = \frac{\varphi_{Z}(-is)}{s}$.
		The final approximation follows by substituting $y=\log (T) $ and $\hat{F}(s) =  {\varphi_Z(-is)}/{s}$ into \eqref{eq:inversion-of-Laplace-transform} and plugging the resulting expression of $F_Y(y)$ into \eqref{eq:F2-FY-relationship}. The proof for the $F_3$ distribution is obtained similarly by setting $y=-\log\left(1-x^{1/\tau}\right)$.
	\end{proof}
	
	\begin{remark}
		The approximations in Proposition \ref{proposition:F2-F3-approximation-inverse} have a singularity at $T = 1$ when computing the $F_2$ distribution and at $x \in \{0,1\}$ when computing the $F_3$ distribution. This, however, poses no problems, since $F_2 (T=1) = 1$,  $F_3 (x=0,\tau) = 1$, and $F_3 (x=1,\tau) = 0$ by definition. Therefore, there is no need to apply  Proposition \ref{proposition:F2-F3-approximation-inverse} for these specific values.
	\end{remark}
	
	\subsection{Approximation of the $F_2$ and $F_3$ Distributions by Using the Beta Distribution}
	
	The approximations of the $F_2$ and $F_3$ distributions provided in Proposition \ref{proposition:F2-F3-approximation-inverse} require the computation of the characteristic function $\varphi_Z(\cdot)$. If the SINR-based coverage criterion is considered, $\varphi_Z(\cdot)$ is available only in integral form, as shown in \eqref{eq:characteristic-function-SINR}. An alternative approach to approximate the $F_2$ and $F_3$ distributions is to utilize the beta distribution, which is shown in \cite{haenggi2015meta} to offer accurate estimates for distributions whose support lies in the entire range $[0,1]$. 
	\begin{proposition}\label{proposition:F2-F3-approximation-Beta}
		Let $\mu$ and $\nu$ be be the mean and variance of $\pcovcon$, respectively, given as follows:
		\begin{align}
			\mu &\mydef \e[\pcovcon] = \varphi_{Z}(-i), \\
			\nu &\mydef \e\left[(\pcovcon)^2\right] - \e\left[\pcovcon\right] = \varphi_{Z}(-2i) - \varphi_{Z}(-i)
		\end{align}
		As for the SINR-based coverage, the $F_2$ and $F_3$ distributions can be approximated as:
		\begin{equation}
			F_2(T) \approx \frac{B(1/T;a,b)}{B(a,b)}, \qquad
			F_3(x,\tau) \approx \frac{B\left(1-x^{1/\tau},a,b\right)}{B(a,b)}. \label{eq:F2-F3-approximation-Beta}
		\end{equation}
		where $B(a,b) = {\Gamma(a)\Gamma(b)}/{\Gamma(a+b)}$ is the Beta function, $B(x,a,b) = \int_{0}^{x} t^{a-1}(1-t)^{b-1}dt$ is the incomplete Beta function, $b = \frac{\mu(1-\mu)^2}{\nu} - (1-\mu)$, and  $a = \frac{\mu b}{1-\mu}$.
	\end{proposition}
	\begin{proof}
		From \eqref{eq:Delay-distribution-F2} and \eqref{eq:Delay-distribution-F3}, the $F_2$ and $F_3$ distributions can be written as $F_2(T) 
		= \pr\left[\pcovcon \leq {1}/{T}\right]$ and $F_3(x,\tau) 
		= \pr\left[\pcovcon \leq 1-x^{1/\tau}\right]$, respectively.
		Under the SINR-based coverage criterion, $\pcovcon$ lies in the  range $[0,1]$. Thus, $F_2(T)$ can be approximated with a Beta distribution whose shape parameters $a$ and $b$ are obtained by matching the mean $\mu$ and variance $\nu$ of $\pcovcon$, i.e., $a = \mu b/(1-\mu)$ and $b = \mu(1-\mu)^2/\nu - (1-\mu)$.
	\end{proof}
	
	\begin{remark}
		As remarked in \cite{haenggi2018efficient}, the Beta approximation can be applied only if the original distribution falls within
		the class of Beta distributions, e.g., the domain of the random variable is the entire interval $[0,1]$. As shown in \cite{wang2019meta}, under the assumption of SIR+ASNR-based coverage, on the other hand, the conditional coverage probability does not necessarily fall within this class. Therefore, the Beta approximation may not be applicable.
	\end{remark}

	\begin{remark}
		The Beta approximation in Proposition \ref{proposition:F2-F3-approximation-Beta} relies on matching the shape of the Beta distribution to that of the original $F_2$ and $F_3$ distributions. This is different from the Euler-sum inversion method in Proposition \ref{proposition:F2-F3-approximation-inverse}, which is an efficient numerical computation of the original distribution with a known truncation error. Therefore, when both approaches are applicable, one is encouraged to use the Euler-sum inversion method for a better accuracy.
	\end{remark}

	\subsection{Approximation of the $F_1$ Distribution via the $F_3$ Distribution}
	The $F_1$ and $F_3$ distributions are related to each other through the following integral relation:
	\begin{equation}
		F_1(\tau) 
		\mystepA  \int_0^{\infty} \pr_\Phi\left[\pr_h \left[\Delta_\Phi > \tau~|~\Phi\right] > x\right]dx  
		=  \int_0^{1} F_3(x,\tau)dx
	\end{equation}
	where $(a)$ follows from the definition of the $F_1$ distribution, i.e., $F_1(\tau) = \e_\Phi\left[\pr_h \left[\Delta > \tau~|~\Phi\right]\right]$, and writing the outer expectation operator as the integral of the CDF. Let $p =\{p_i\}$ be an ordered sequences of real numbers between 0 and 1, i.e., $0 = p_0 < p_1 < \ldots < p_{n-1} < p_n = 1$. Also, let ${\epsilon} =\{\epsilon_i\}$ be a sequence that satisfies $p_{i-1} < \epsilon_i \leq p_i, \forall i\in\{1,2,\ldots,n\}$. The $F_1$ distribution can be expressed as a limit of a sum of the $F_3$ distribution, i.e.,
	\begin{equation}\label{eq:F1-as-sum-of-F3}
		F_1(\tau) = \lim_{n\to\infty} \sum_{i=1}^n (p_i - p_{i-1})F_3(\epsilon_i,\tau).
	\end{equation}
	
	\begin{remark}
		By using \eqref{eq:F1-as-sum-of-F3} with a sufficiently large value of $n$, we can compute the $F_1$ distribution efficiently.
		By using this technique, we avoid working with large values of the binomial coefficients that appear in Theorem \ref{thm:F1-distribution-general}, which usually leads to numerical issues in evaluating the distribution.
	\end{remark}

	\section{Numerical Results}\label{sec:Numerical-Results}
	
	\begin{table}[!t]
		\footnotesize
		\centering
		\caption{Simulation setup.}
		\label{table:simulation_setup}
		\newcommand{\tabincell}[2]{\begin{tabular}{@{}#1@{}}#2\end{tabular}}
		\begin{tabular}{|l|l|} \hline
			Parameter & \hspace{0.5cm} Value\\ 
			\hline
			$\alpha$ & $4$\\
			$f_c$	& $2.1$ GHz \\
			$K$ \quad & $(4\pi f_c/3\cdot10^{8})^2$ \\
			$N_0$ & $-174$ dBm/Hz \\
			$B_W$   & $200$ MHz\\
			$P$   & $43$ dBm\\
			$\textup{R}_{\textup{MT}}$ & $50$ m\\
			$\lambdaMT$ & $1/(\pi \textup{R}^2_{\textup{MT}})$   \\
			$\gamma_A$ & $12.5$ dBm  \\
			\hline
		\end{tabular}
	\end{table}
	
	In this section, we show numerical results to validate the proposed analytical frameworks for computing the local delay and the delay distributions, as well as to substantiate the obtained findings on the packet loss probability. Unless otherwise stated, the simulation setup is summarized
	in Table \ref{table:simulation_setup}. 
	The simulation results presented in this section are obtained through Monte Carlo simulations, by generating $5,000$ realizations of the BS and MT, as well as simulating $5,000$ packet transmissions for each point process (spatial) realization.

	\subsection{Local Delay}
	
	\begin{figure}[!t]
		\centering
		\begin{center}
			\includegraphics[width=0.6\linewidth]{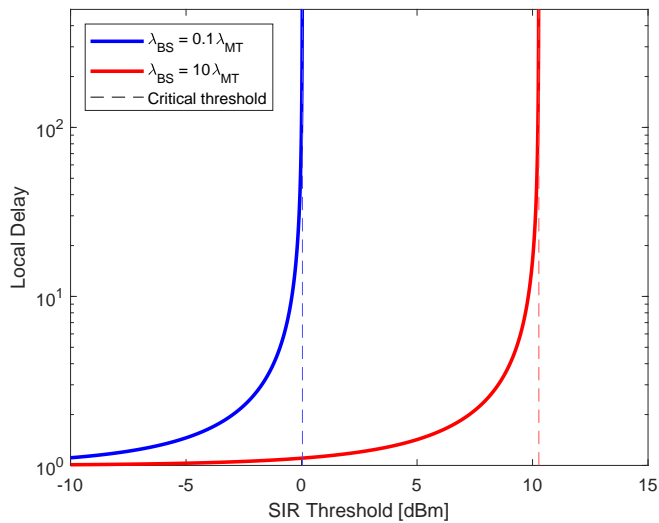}
			\caption{Local delay as a function of the decoding threshold (SIR-based coverage criterion).}
			\label{fig:LocalDelaySIR}
		\end{center} \vspace{-1cm}
	\end{figure}
	
	Fig. \ref{fig:LocalDelaySIR} validates the analytical framework for the local delay that corresponds to the SIR-based coverage criterion, as a function of the threshold $\gamma_D$ for different values of the BS density $\lambdaBS$. The theoretical results are computed by using \eqref{eq:local-delay-SIR}. From the figure, we can see that the local delay increases with the BS density and the decoding threshold. This reveals that, given a spatial realization of BSs and MTs, a typical MT is more likely to successfully delivers its data packets when the number of BSs increases due to shorter distances, on average, despite the MT is likely to experience more interference due to higher number of interfering BSs. We observe, in addition, that the local delay tends to infinity near the critical threshold given in \eqref{eq:critical-treshold}.

	\subsection{$F_1$ Distribution}
	
	\begin{figure}[!t]
		\begin{center}
			\begin{subfigure}{0.48\columnwidth}
				{\includegraphics[width=0.98\linewidth]{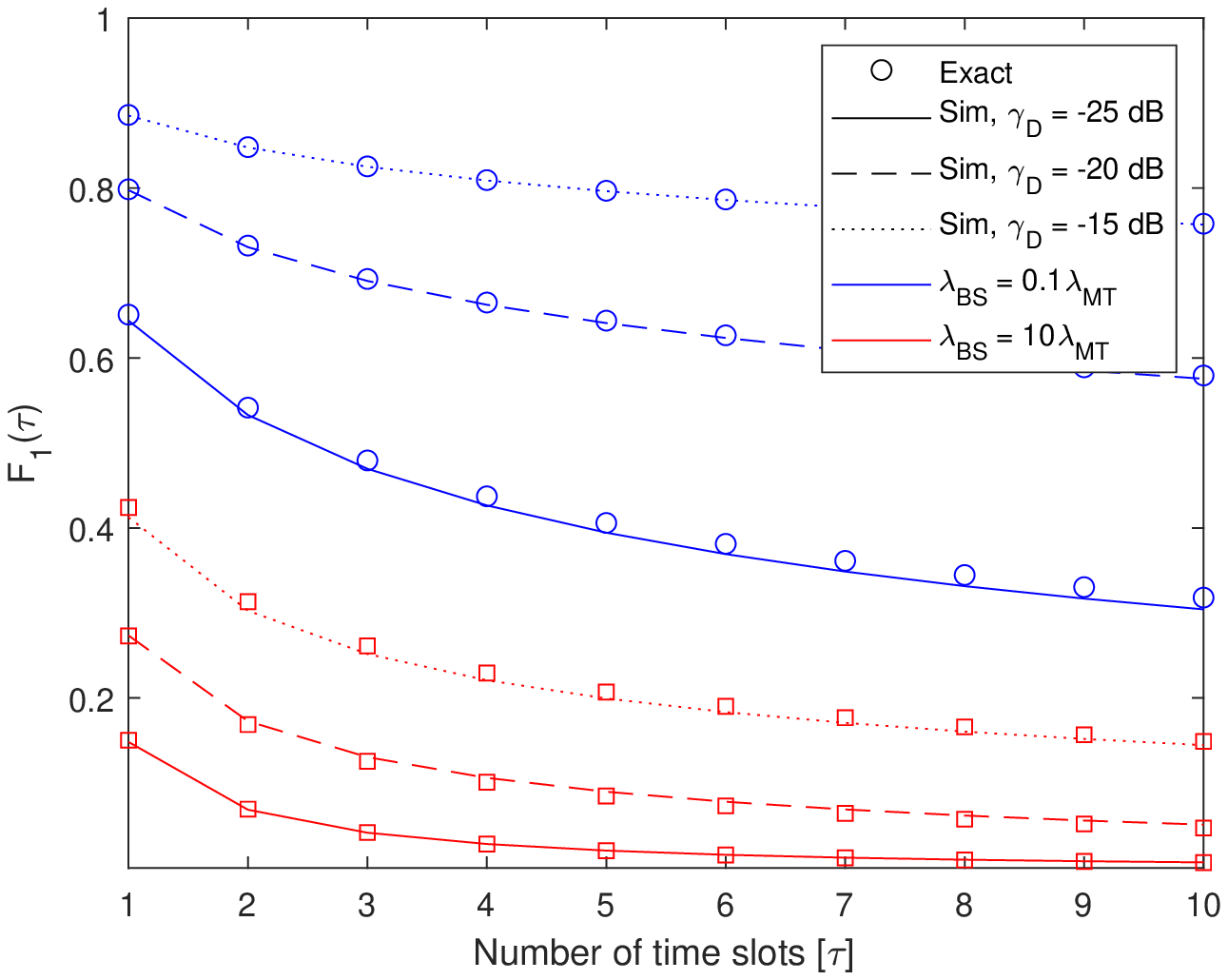}}
				\caption{SIR-based coverage.}\label{fig:F1_SIR}
			\end{subfigure}
			\begin{subfigure}{0.48\columnwidth}
				{\includegraphics[width=0.98\linewidth]{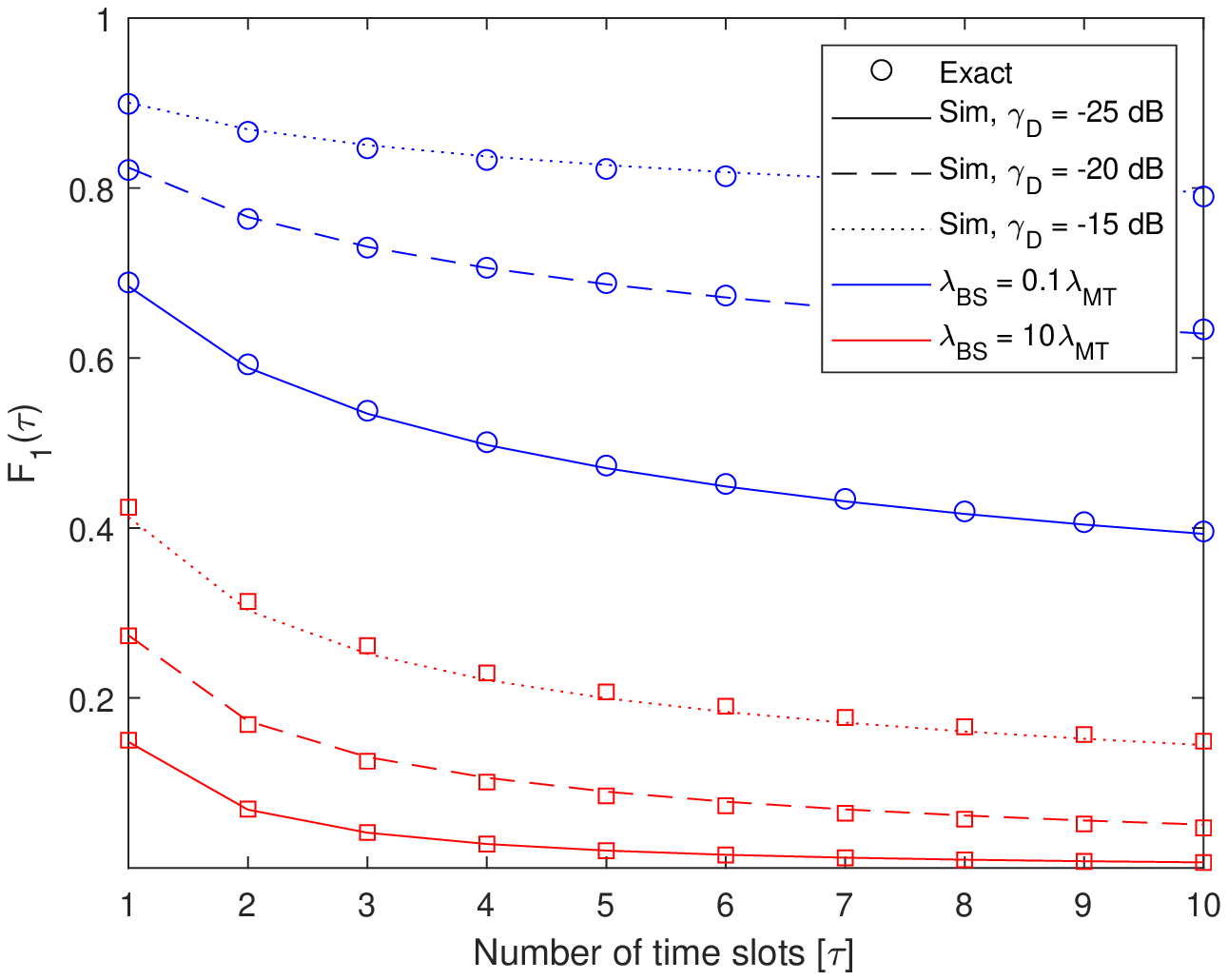}}
				\caption{SINR-based coverage.}\label{fig:F1_SINR}
			\end{subfigure}
			\begin{subfigure}{0.48\columnwidth}
				{\includegraphics[width=0.98\linewidth]{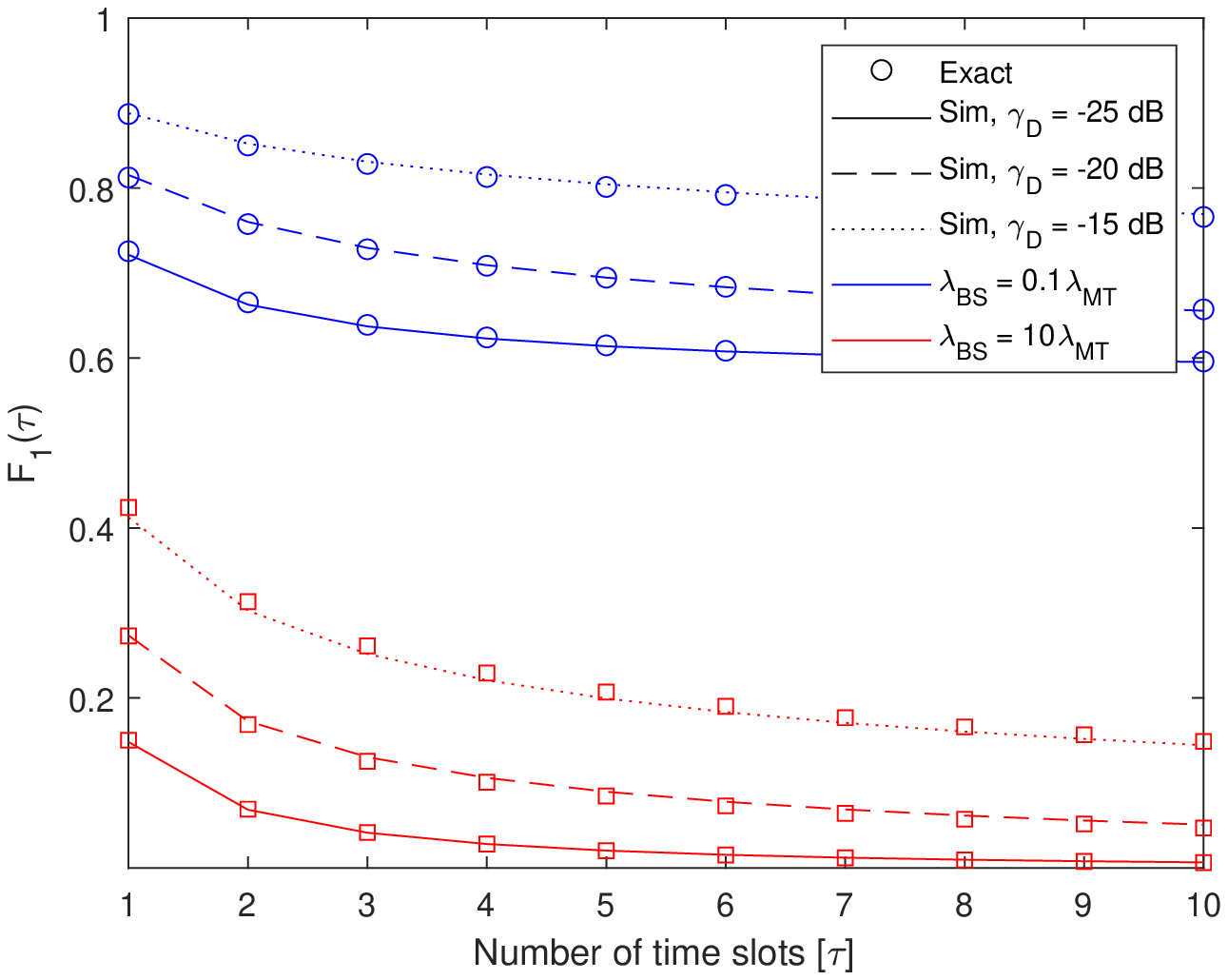}}
				\caption{SIR+ASNR-based coverage.}\label{fig:F1_SIR_ASNR}
			\end{subfigure}
			\caption{$F_1$ distribution for different coverage criteria.}
			\label{fig:F1}
		\end{center} \vspace{-1cm}
	\end{figure}
	
	\begin{figure}[!t]
		\begin{center}
			\includegraphics[width=0.6\linewidth]{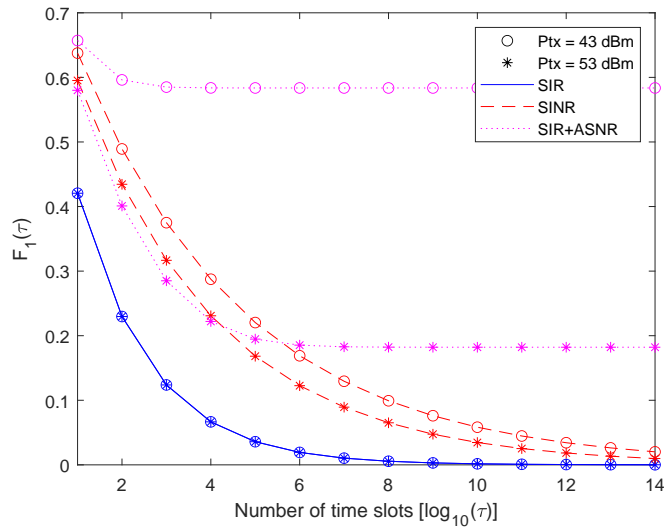}
			\caption{$F_1$ distribution for different coverage criteria as a function of number of time slots, decoding threshold, and BS density.}
			\label{fig:F1_asymptote}
		\end{center} \vspace{-1cm}
	\end{figure}
	
	Fig. \ref{fig:F1} validates the analytical frameworks for the $F_1$ distribution, as a function of number of time slots, decoding threshold, and BS density. The theoretical results are computed as a function of the time slots $\tau$ by using the analytical formulations in \eqref{eq:F1-distribution-SIR}, \eqref{eq:F1-distribution-SINR}, and \eqref{eq:F1-distribution-SIR+ASNR}. From the figure, we can see that the analytical formulations agree with the simulation results. We also observe that the delay performance is better for smaller values of the threshold and higher values of the BS density. Fig. \ref{fig:F1_asymptote} shows the $F_1$ distribution as $\tau$ tends to infinity, i.e., it illustrates the packet loss probability. As evident from the figure, the probability goes to zero as $\tau$ goes to infinity when the SIR and SINR-based coverage criteria are considered, while it goes to a non-zero value when the SIR+ASNR-based coverage criterion is utilized. We also observe that the delay performance improves when the transmit power $P$ increases.

	\subsection{$F_2$ Distribution}
	
	Fig. \ref{fig:F2} validates the computational framework for the $F_2$ distribution as a function of the number of time slots $T$, the decoding threshold $\gamma_D$, and the BS density $\lambdaBS$. The approximated results are obtained, for the SIR-based and SIR+ASNR-based coverage criteria, by using the inverse Laplace transform method and, for the SINR-based coverage criterion, the Beta approximation. From the figures, we can observe that the delay performance is better for lower values of the threshold and for higher values of the BS density. We also observe that, in a fully-loaded scenario (i.e., when $\lambdaBS \ll \lambdaMT$), the impact of the noise (for the SINR-based and SIR+ASNR-based coverage criterion) is not apparent. This is because, in this case, the system is nearly interference-limited, i.e., the interference is a more dominant factor as compared to the noise. On the other hand, the impact of the noise is clearly seen in lightly-loaded scenario (i.e., when $\lambdaBS \gg \lambdaMT$), since several BSs may be turned off.
	
	\begin{figure}[!t]
		\begin{center}
			\begin{subfigure}{0.48\columnwidth}
				{\includegraphics[width=0.98\linewidth]{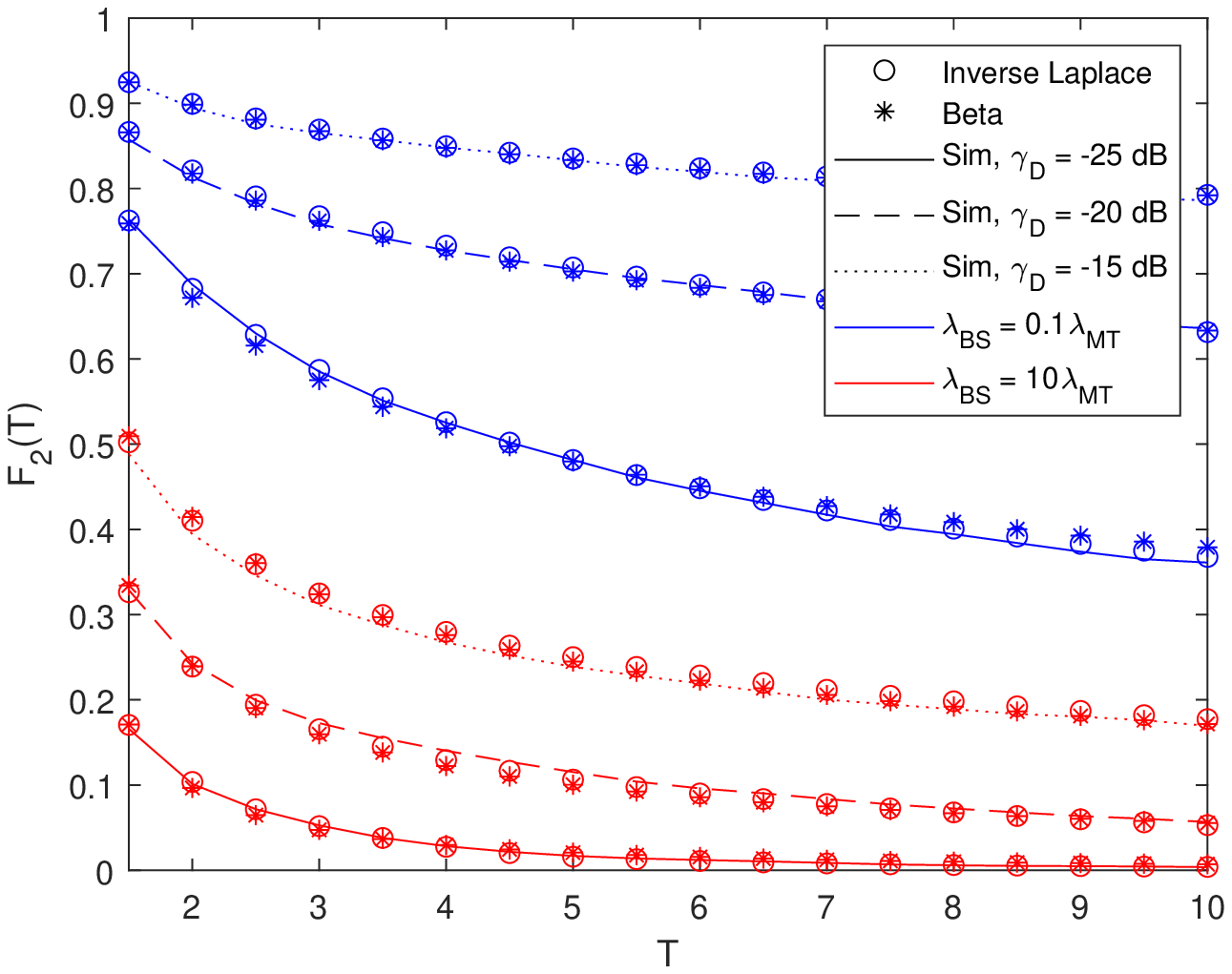}}
				\caption{SIR-based coverage.}\label{fig:F2_SIR}
			\end{subfigure}
			\begin{subfigure}{0.48\columnwidth}
				{\includegraphics[width=0.98\linewidth]{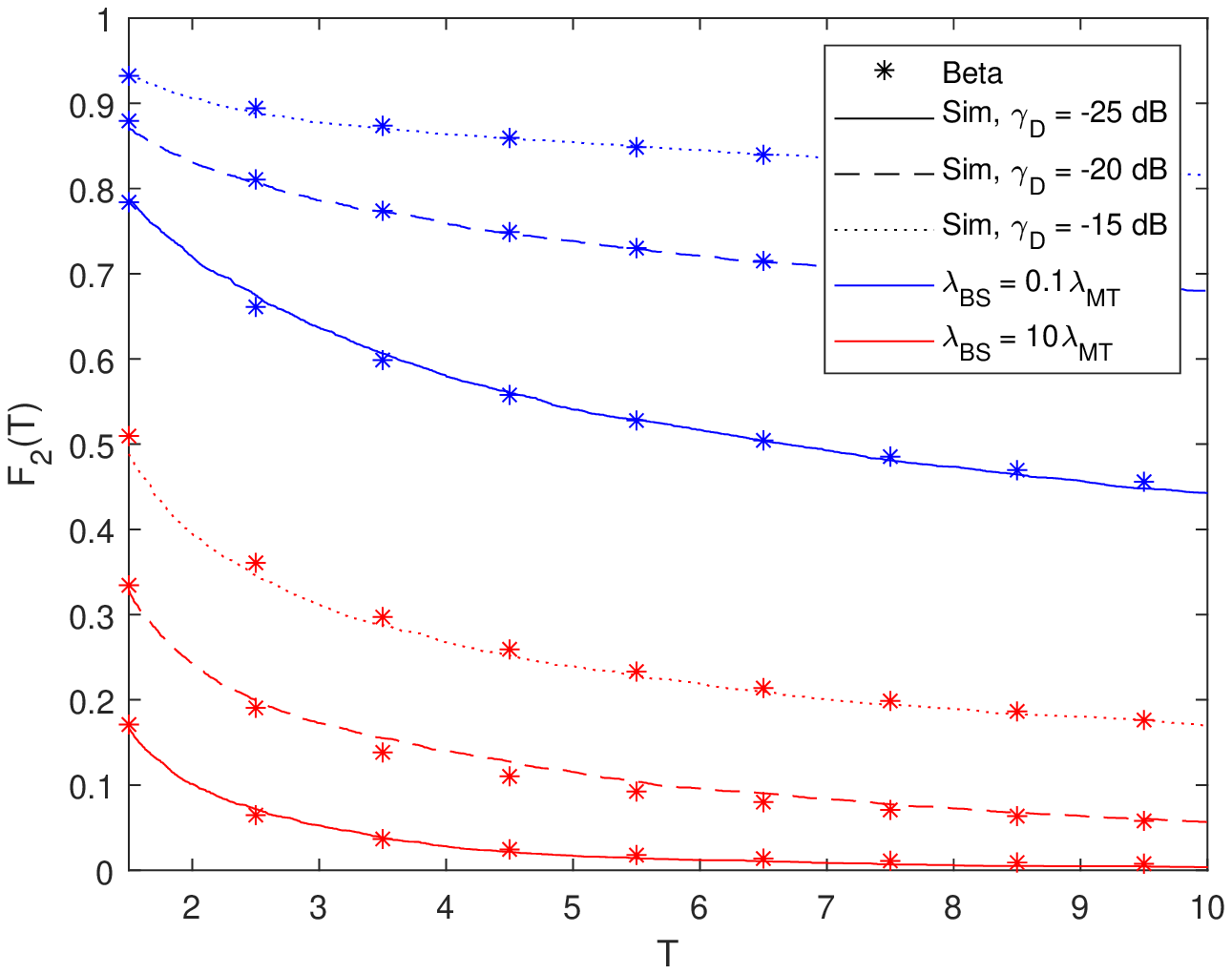}}
				\caption{SINR-based coverage.}\label{fig:F2_SINR}
			\end{subfigure}
			\begin{subfigure}{0.48\columnwidth}
				{\includegraphics[width=0.98\linewidth]{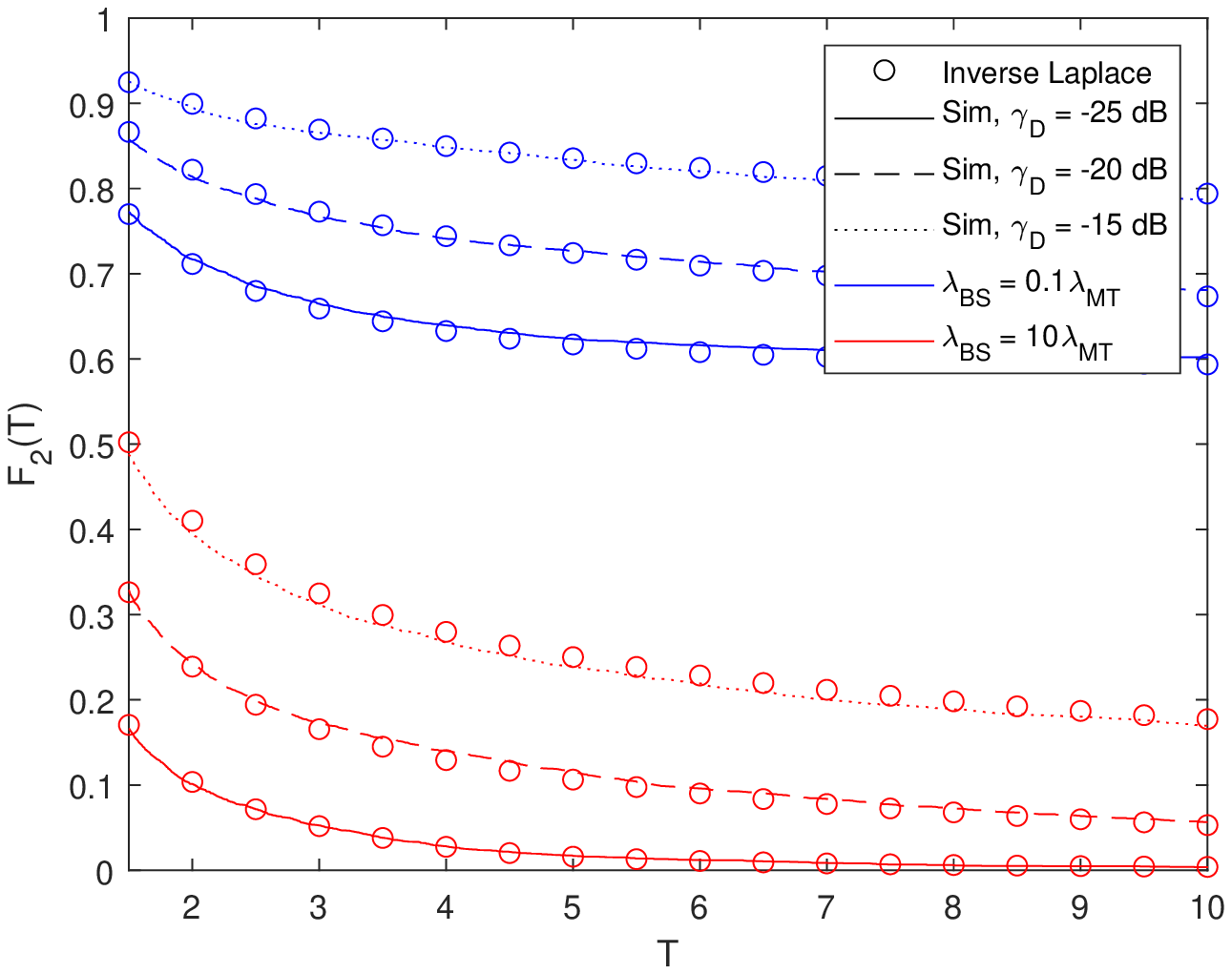}}
				\caption{SIR+ASNR-based coverage.}\label{fig:F2_SIR_ASNR}
			\end{subfigure}
			\caption{$F_2$ distribution for different coverage criteria.}
			\label{fig:F2}
		\end{center} \vspace{-1cm}
	\end{figure}

	\subsection{$F_3$ Distribution}
	
	Fig. \ref{fig:F3} validates the computational framework for the $F_3$ distribution as a function of the number of time slots $\tau \in \{5,10\}$ and the decoding threshold $\gamma_D$ when the BS density is $\lambdaBS = 0.1 \lambdaMT$. The approximated results are computed by using the inverse Laplace transform method (for the SIR-based and SIR+ASNR-based coverage criteria) and the Beta approximation (for the SINR-based coverage criterion). From the figures, we can see that the analytical formulations agree with the simulation results. We observe that the delay performance is better for lower values of the thresholds and for higher values of the BS density. When the SIR+ASNR-based coverage is considered, in addition, the $F_3$ distribution tends to a non-zero value when $x$ tends to zero.
	
	\begin{figure}[!t]
		\begin{center}
			\begin{subfigure}{0.48\columnwidth}
				{\includegraphics[width=0.98\linewidth]{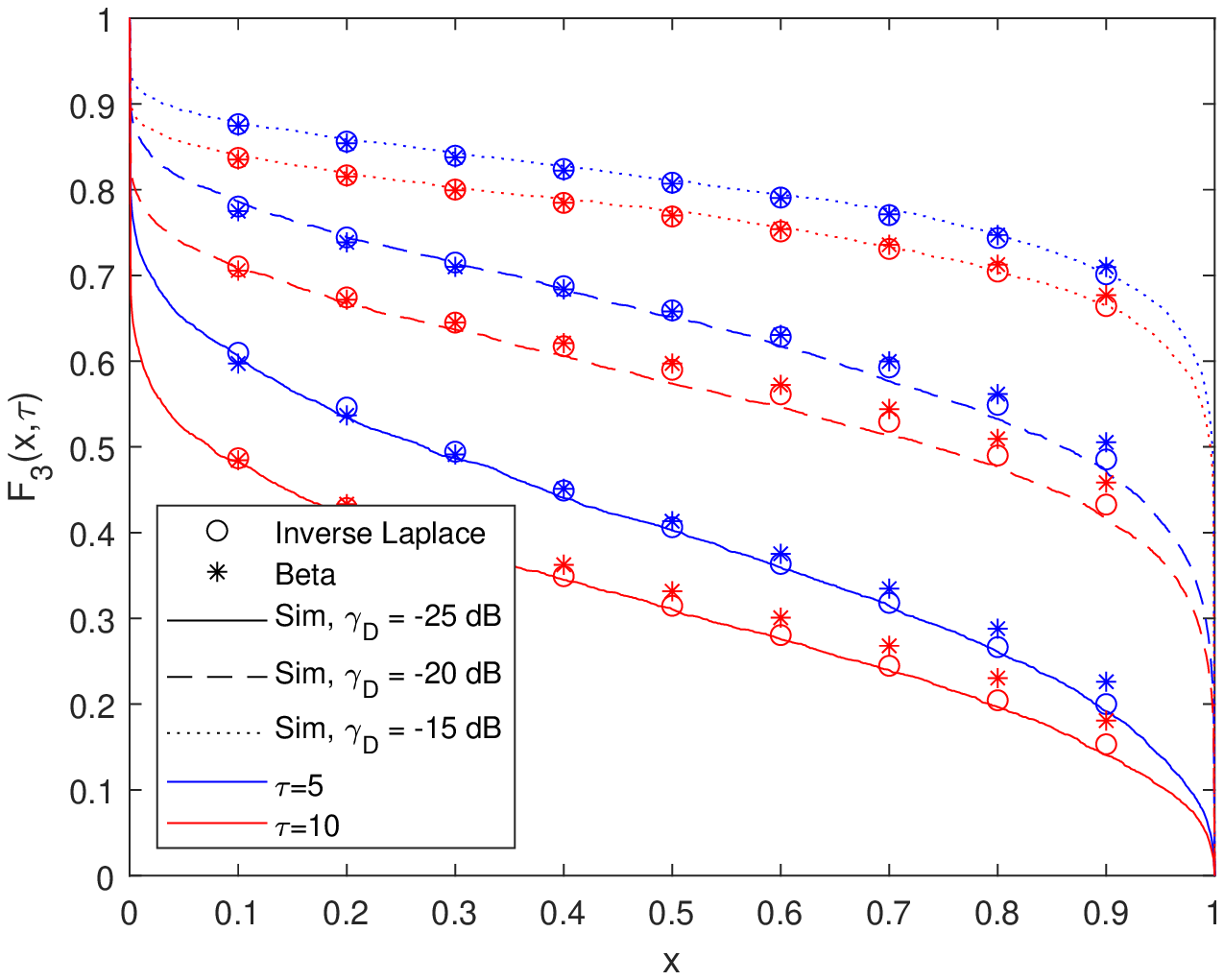}}
				\caption{SIR-based coverage.}\label{fig:F3_SIR}
			\end{subfigure}
			\begin{subfigure}{0.48\columnwidth}
				{\includegraphics[width=0.98\linewidth]{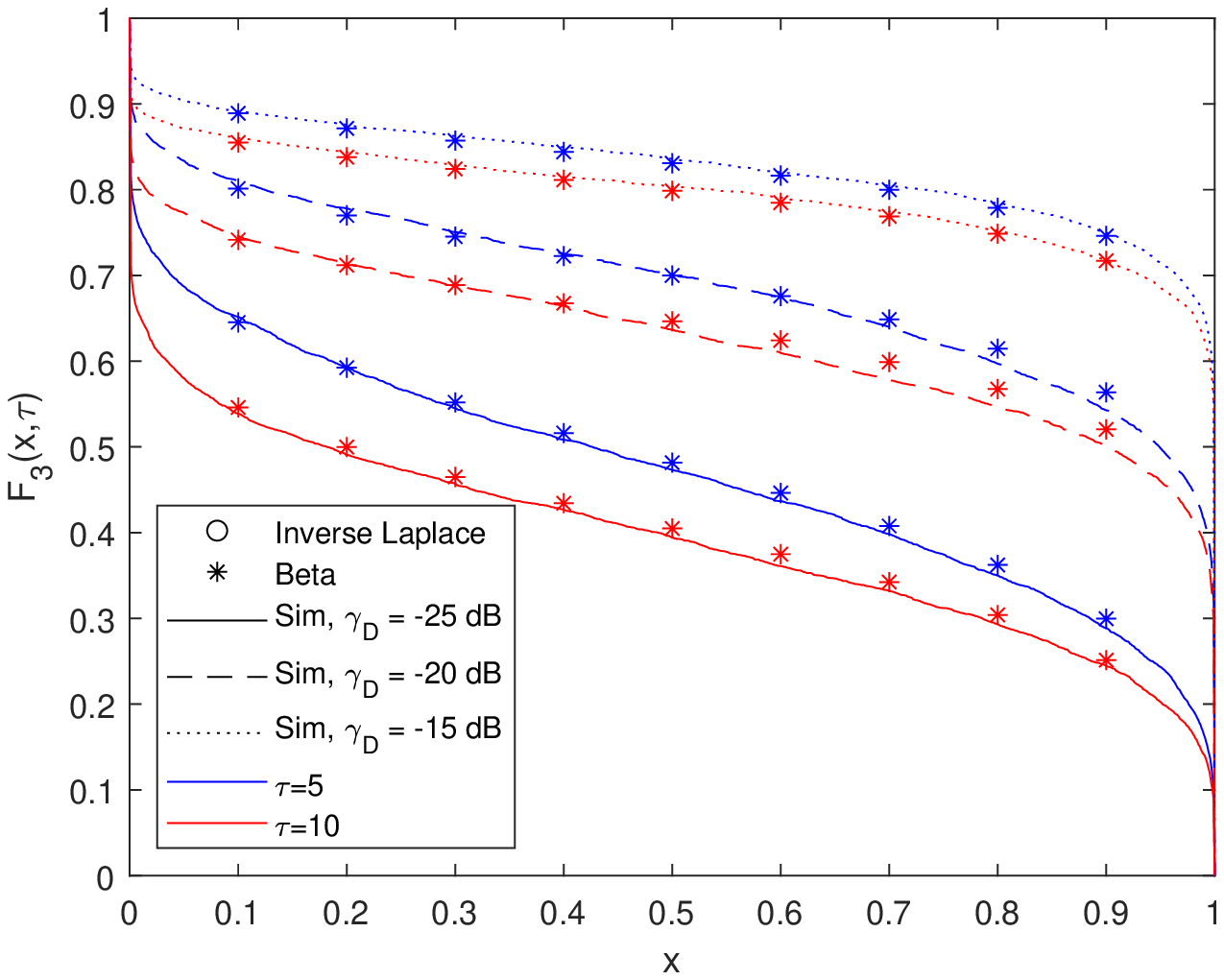}}
				\caption{SINR-based coverage.}\label{fig:F3_SINR}
			\end{subfigure}
			\begin{subfigure}{0.48\columnwidth}
				{\includegraphics[width=0.98\linewidth]{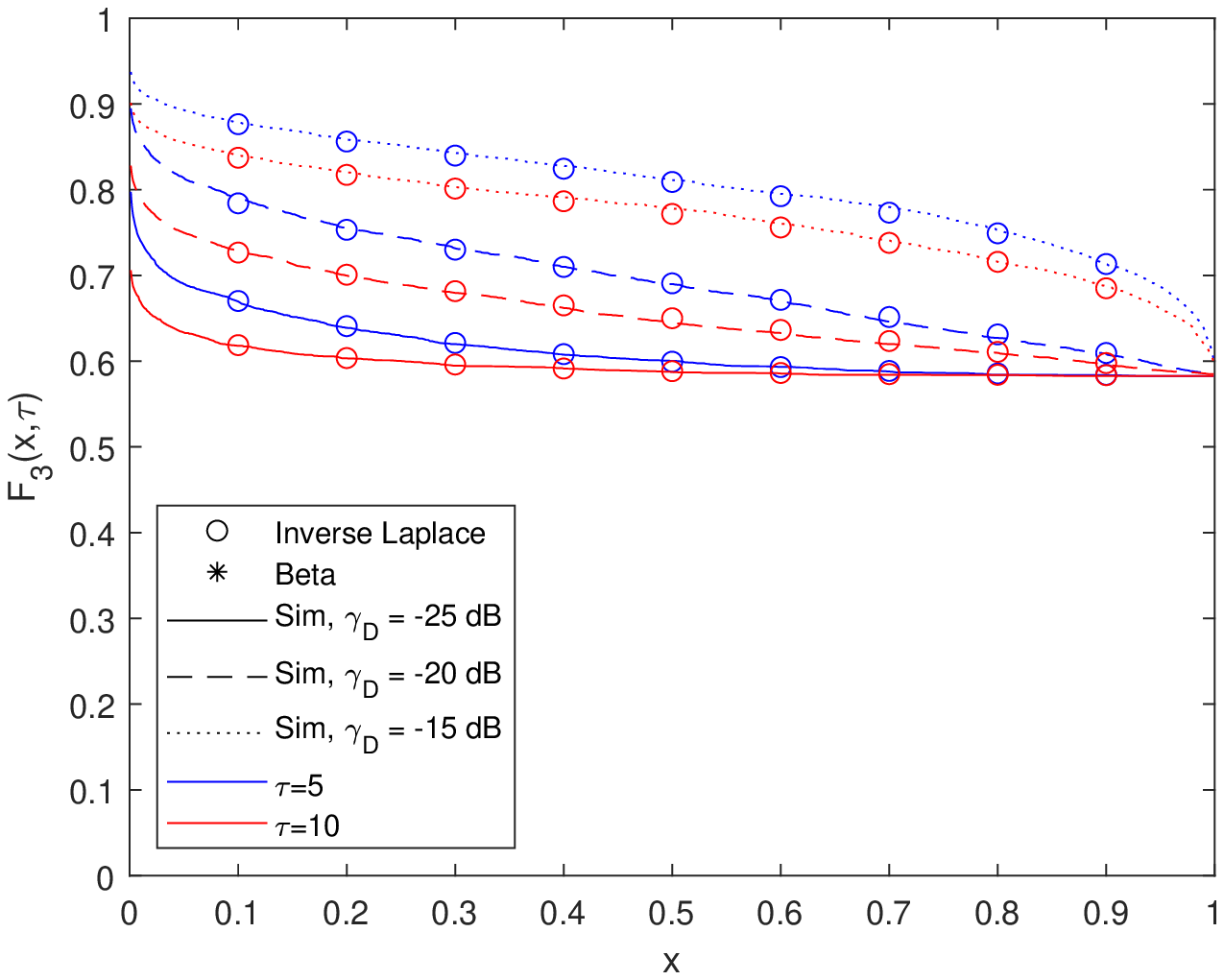}}
				\caption{SIR+ASNR-based coverage.}\label{fig:F3_SIR_ASNR}
			\end{subfigure}
			\caption{$F_3$ distribution for different coverage criteria.}
			\label{fig:F3}
		\end{center} \vspace{-1cm}
	\end{figure}

	\section{Conclusion} \label{sec:Conclusion}

	In this paper, we have introduced three different distributions to characterize the performance of large-scale cellular networks in terms of data packet delay. The proposed distributions provide complete information on the statistics of the delay, which overcomes the definition of local delay typically used in wireless communications. We have proposed analytical frameworks for evaluating the distributions for three different coverage criteria, namely the SIR-based, SINR-based, and SIR+ASNR-based coverage criteria. Several numerical approximations based on the inversion method, the Riemann sum, and the Beta distribution have been employed to circumvent some inherent numerical difficulties in computing the obtained analytical formulations for some case studies, as well as to render the computation of the proposed analytical methods more efficient. Also, we have analyzed the asymptotic behavior for one of the proposed delay distributions, which provides information on the network packet loss probability. Finally, through extensive simulations, we have shown that the obtained analytical formulations and approximations are accurate compared to Monte Carlo simulation results, and have studied the impact on the delay performance of several design parameters, such as the decoding threshold, the transmit power, and the network deployment density. The proposed approaches can facilitate the analysis and optimization of cellular networks subject to reliability constraints on the network packet delay that are not restricted to the local (average) delay, e.g., for delay sensitive applications.

	\begin{appendices}
		
		\section{Proof of Corollary \ref{corollary:delay-F1-relation}}\label{appendix:proof-of-delay-F1-relation}
		
		By virtue of the linearity of the expectation operator, we can rewrite the local delay in \eqref{eq:local-delay-def} as follows:
		\begin{align}
			D 
			= \e_{\Phi}\left[\e_h\left[\Delta_{\Phi}\right]\right]
			= \e_{\Phi} \left[\sum_{\tau=0}^{\infty} \tau \pr_h\left[\Delta_{\Phi} = \tau~|~\Phi\right]\right]
			= \sum_{\tau=0}^{\infty} \tau \e_{\Phi} \left[ \pr_h\left[\Delta_{\Phi} = \tau~|~\Phi\right]\right].
		\end{align}
		The proof follows from the definition of probability mass function (PMF) of the discrete random variable $\Delta_{\Phi}$, i.e., $\e_{\Phi}\left[\pr\left[\Delta_{\Phi} = \tau~|~\Phi\right]\right] = F_1(\tau) - F_1(\tau-1)$.

		\section{Proof of Lemma \ref{lemma:coverage-probability-general-formula}}
		\label{appendix:proof-of-coverage-probability-general-formula}
		
		When the $\sinr$-based coverage criterion is considered, the CCP is defined as follows:
		\begin{align}
			\pcovcon 
			&= \pr_h\left[\frac{{Ph_{0}}/{\ell(r_0)}}{I+W} > \gamma~|~\Phi\right] \\
			&\mystepA \e_{h_i}\left[e^{-\frac{\gamma(W + I)K{r_0}^\alpha}{P}}~\Bigr|~\Phi\right] 
			\mystepB e^{-\frac{\gamma WK{r_0}^\alpha}{P}}\e_{h_{i}}\left[\prod_{i\in \PhiI}e^{-\gamma h_{i}\left(\frac{r_0}{r_i}\right)^\alpha}~\Bigr|~\Phi\right] \nonumber\\
			&\mystepC e^{-\frac{\gamma WK{r_0}^\alpha}{P}}\prod_{i\in \PhiI}\e_{h_{i}}\left[e^{-\gamma h_{i}\left(\frac{r_0}{r_i}\right)^\alpha}~\Bigr|~\Phi\right] 
			\mystepD e^{-\frac{\gamma WK{r_0}^\alpha}{P}}\prod_{i\in \PhiI} \left(1+\gamma\left(\frac{r_0}{r_i}\right)^\alpha\right)^{-1}
		\end{align}
		where $(a)$ comes from the definition of $\sinr$ and $\pcovcon$ from the assumption that the random variable $h_0$ is exponentially distributed with unit mean, $(b)$ is obtained by substituting the total interference $I$ and by taking into account that the term $\exp(-\gamma W \ell(r_0)/P)$ is independent of $I$, $(c)$ follows from the fact that the fading coefficients $h_i$ are mutually independent, and $(d)$ follows because $h_i$ is an exponentially distributed random variable with unit mean.
		
		The CCP for the SIR-based coverage criterion can be derived by setting $W=0$. When the SIR+ASNR-based coverage criterion is considered, the CCP is:
		\begin{align}
			\pcovcon 
			&= \pr_h\left[h_{0} > \frac{\gamma K{r_0}^\alpha I}{P},~ W \leq \frac{P}{K{r_0}^\alpha\theta}~\Bigr|~\Phi\right]\nonumber\\
			&\mystepE \pr_h\left[h_{0} > \frac{\gamma K{r_0}^\alpha I}{P},~\Bigr|~\Phi\right] \mathbbm{1}\left({r_0 \leq \left(\frac{P}{KW\theta}\right)^{1/\alpha}}\right).
		\end{align}
		where $(e)$ holds because $W$ is independent of $h_0$ and $I$. The proof follows using similar steps as for the $\sinr$-based coverage and by noting that $\pr\left[W \leq {P}/({K{r_0}^\alpha\theta})|\Phi\right] = \mathbbm{1}\left({r_0 \leq \left(\frac{P}{KW\theta}\right)^{1/\alpha}}\right)$.

		\section{Proof of Theorem \ref{thm:local-delay-general}}\label{appendix:local-delay-general}
		
		By substituting  \eqref{eq:coverage-probability-general-formula} into \eqref{eq:local-delay-def}, and 	separating the conditioning on $\Phi$ into 
		$r_0$ and 
		$\PhiI$, we have:
		\begin{align}\label{eq:Delay-SINR-step1}
			D
			&= \e_{\PhiI,r_0}\left[\left(\mathcal{G}\left(r_0\right)\right)^{-1}\prod_{i\in \PhiI} \left(1+\gamma\left(\frac{r_0}{r_i}\right)^\alpha\right)\right] \nonumber\\
			&\mystepA \e_{r_0}\left[\left(\mathcal{G}\left(r_0\right)\right)^{-1}\e_{\PhiI}\left[\prod_{i\in \PhiI} \left(1+\gamma \left(\frac{r_0}{r_i}\right)^\alpha\right)\right]\right]
		\end{align}
		where $(a)$ follows from the fact that the term $\left(\mathcal{G}\left(r_0\right)\right)^{-1}$ is independent of $\PhiI$. According to the probability generating functional (PGFL) theorem of a PPP  \cite{chiu2013stochastic}, for a general function $f(x)$ and a PPP $\Psi$ over $\mathbb{R}^2$, we have: 
		\begin{equation}
			\e\left[\prod_{x\in\Psi} f(x)\right] = \exp\left(-\int_{\mathbb{R}^2}(1-f(x))\lambda(x)dx\right)
		\end{equation}
		where $\lambda(x)$ is the density function of the PPP. By expressing the whole domain $\mathbb{R}^2$ in polar coordinate, we can write $dx = rdrd\phi$ where $\phi$ is the azimuth angle of a particular point with respect to the origin. Using this notation, we obtain:		\begin{align}\label{eq:PGFL}
			\e_{\PhiI}\left[\prod_{i\in \PhiI} \left(1+\gamma\left(\frac{r_0}{r_i}\right)^\alpha\right)\right] 
			&= \exp\left(-\int\limits^{2\pi}_{0}\int\limits^{\infty}_{0}\left(1- \left(1+\gamma\left(\frac{r_0}{r_i}\right)^\alpha\right)\right)\lambda(r_i,\phi)r_idr_id\phi\right)\nonumber\\
			&\mystepB \exp\left(2\pi\lambdaBS \gamma L(\lambdaBS,\lambdaMT) {r_0}^\alpha \int\limits^{\infty}_{r_0} (r_i)^{1-\alpha} dr_i\right)\nonumber\\
			&= \exp\left(\frac{2\pi\lambdaBS\gamma L(\lambdaBS,\lambdaMT) {r_0}^2}{\alpha-2}\right).
		\end{align}
		where $(b)$ is obtained by substituting $\lambda(r_i,\phi)$ with \eqref{eq:intensity-inhomogenous} . The proof follows by substituting \eqref{eq:PGFL} into \eqref{eq:Delay-SINR-step1} and by then calculating the last expectation over $r_0$ with the aid of the PDF in (\ref{eq:PDF-of-r0}).

		\section{Proof of Theorem \ref{thm:F1-distribution-general}}
		\label{appendix:proof-of-F1-distribution-general}
		
		Since $\Delta_\Phi$ is a geometrically distributed random variable, its conditional CCDF is:
		\begin{align}
			\pr_h \left[\Delta_\Phi > \tau~|~\Phi\right] = \left(1-\pcovcon\right)^{\tau} = \sum_{k=0}^{\tau} (-1)^k C_k^\tau \left(\pcovcon\right)^k
		\end{align}
		From (\ref{eq:Delay-distribution-F1}) and by separating the conditioning on $\Phi$ into $r_0$ and $\PhiI$, we obtain:
		\begin{align}
			F_1(\tau) 
			&= \sum_{k=0}^{\tau} (-1)^k C_k^\tau \e_{\Phi}\left[ \left(\pcovcon\right)^k \right] \nonumber\\
			&= \sum_{k=0}^{\tau} (-1)^k C_k^\tau \e_{r_0}\left[\e_{\PhiI}\left[\left(\pcovcon\right)^{k}\right]\right] \\
			&\mystepA \sum_{k=0}^{\tau} C_{k}^{\tau} (-1)^{k} \e_{r_0}\left[\e_{\PhiI}\left[\prod_{i\in\PhiI}\left(1+\gamma\left(\frac{r_0}{r_i}\right)^{\alpha}\right)^{-k}\right]  \left(\mathcal{G}(r_0)\right)^k\right] \nonumber\\
			&\mystepB \sum_{k=0}^{\tau} C_{k}^{\tau}(-1)^{k}\e_{r_0}\left[\left(\mathcal{G}(r_0)\right)^k  e^{-\pi\lambdaBS L(\lambdaBS,\lambdaMT) r_0^2\mathcal{F}\left(k,\alpha,\gamma\right)}\right] \nonumber
		\end{align}
		where $(a)$ is obtained by substituting the CCP given in Lemma \ref{lemma:coverage-probability-general-formula} and $(b)$ follows by applying the PGFL theorem and using the notable integral $\int_{r_0}^{\infty}\left\{1-\left[1+\gamma\left(\frac{r_0}{r_i}\right)^{\alpha}\right]^{-k}\right\}r_idr_i = \frac{1}{2}{r_0}^2 \left({}_2F_1\left(-\frac{2}{\alpha},k;\frac{\alpha-2}{\alpha};-\gamma\right)-1\right)$.
		The proof follows from deconditioning the expectation on $r_0$.
		
		\section{Proof of Corollary \ref{corollary:F1-distribution-exact}} \label{appendix:proof-of-F1-distribution-exact}
		
		As for $F_1^{\textup{sir+asnr}}(\tau)$, we first substitute $\mathcal{G}(r_0) = \mathbbm{1}\left({r_0 \leq \left(\frac{P}{KW\theta}\right)^{1/\alpha}}\right)$ into \eqref{eq:F1-distribution-general} and obtain:
		\begin{align}
			F_1^{\textup{sir+asnr}}(\tau)
			&= 2\pi\lambdaBS \sum_{k=0}^{\tau} C_{k}^{\tau}(-1)^{k} \int_0^\infty \left(\mathbbm{1}\left({r_0 \leq \left(\frac{P}{KW\theta}\right)^{1/\alpha}}\right)\right)^k  e^{-\pi\lambdaBS r_0^2 \mathcal{F}\left(k,\alpha,\gamma\right) }  r_0dr_0 \nonumber\\
			&\mystepA 1 + 2\pi\lambdaBS \sum_{k=1}^{\tau} C_{k}^{\tau}(-1)^{k} \int_0^{\left(\frac{P}{KW\theta}\right)^{\frac{1}{\alpha}}}  e^{-\pi\lambdaBS r_0^2 \mathcal{F}\left(k,\alpha,\gamma\right) }  r_0dr_0
		\end{align}
		where $(a)$ follows from two steps: (i) splitting the sum for $k=0$ and for $k\geq 1$ as well as using the fact that $\left(\mathbbm{1}\left({r_0 \leq \left(\frac{P}{KW\theta}\right)^{1/\alpha}}\right)\right)^0 = 1$, and
		(ii) by taking into account that, for $k \geq 1$, $\left(\mathbbm{1}\left({r_0 \leq \left(\frac{P}{KW\theta}\right)^{{1/\alpha}}}\right)\right)^k$ is equal to $1$ if $r_0 \leq \left(\frac{P}{KW\theta}\right)^{1/\alpha}$ and it is equal to $0$ otherwise. Computing the final integral completes the proof.

		\section{Proof of Lemma \ref{lemma:positivity-of-pcovcon}} \label{appendix:positivity-of-pcovcon}
		
		We start the proof by first proving the following intermediate result: if $\{a_i\}$ is an infinite sequence such that $0<a_1<a_2<\ldots<1$, then the following holds: 
		\begin{equation}\label{eq:intermediate-result}
			\prod_{i=1}^{\infty} a_i > 0 \iff \sum_{i=0}^{\infty} (1-a_i) < \infty.
		\end{equation}
	
		First, we note that, due to the equivalence of implication, \eqref{eq:intermediate-result} is equivalent to $\prod_{i=1}^{\infty} a_i \leq 0 \iff \sum_{i=0}^{\infty} (1-a_i) \geq \infty$. Also, since $a_i > 0$ for all $i$, the latter implication is equivalent to:
		\begin{equation}\label{eq:intermediate-result-2}
			\prod_{i=1}^{\infty} a_i = 0 \iff \sum_{i=0}^{\infty} (1-a_i) = \infty.
		\end{equation}
		
		To prove \eqref{eq:intermediate-result}, it is sufficient to prove \eqref{eq:intermediate-result-2}. We consider two cases.
		\subsubsection{Case 1: $\{a_i\}$ converges to $0 < M < 1$}
		Assume that $\{a_i\}$ converges to $M$ where $0 < M < 1$, that is, $\lim_{i \to \infty} a_i = M < 1$. Since $a_i < M$ for all $i$, we have 
		$\prod_{i=1}^{\infty} a_i < \prod_{i=1}^{\infty} M \mystepA 0$
		where $(a)$ follows from the fact that $0< M < 1$. 
		Since $\lim_{i \to \infty} a_i = M < 1$, we have $\lim_{i \to \infty} (1-a_i) > 0$. Due to divergence test, this implies $\sum_{i=0}^{\infty} (1-a_i) = \infty$. This implies that \eqref{eq:intermediate-result-2} (and hence \eqref{eq:intermediate-result}) holds if $\{a_i\}$ converges to $M < 1$.
		
		\subsubsection{Case 2: $\{a_i\}$ converges to $M = 1$}
		Assume that $\{a_i\}$ converges to $M=1$, that is, $\lim_{i \to \infty} a_i = 1$. We have
		$\prod_{i=1}^{\infty} a_i = 0 \iffB \log\left(\prod_{i=1}^{\infty} a_i\right) = \log(0) 
		\iffC \sum_{i=1}^{\infty} \log(a_i) = -\infty$,
		where $(b)$ follows from taking the logarithm of both sides of the equation and $(c)$ follows from the fact that the logarithm of a product is equal to the sum of the logarithm as well as $\log(0) = -\infty$. 
		Without loss of generality, let us assume that  $a_1 > K$ for some $0<K<1$. Due to the fact that $\{a_i\}$ is an ordered sequence, we have $a_i > K$ for all $i$ and, hence $
		1/a_i < 1/a_1 = 1/K$.
		Multiplying both sides with $a_i - 1 < 0$, we have
		$\frac{{{a_i} - 1}}{a_i} > \frac{{{a_i} - 1}}{K}$.
		According to \cite[Eq. (1)]{topsok2006some}, for any $a_i>0$, we obtain $\frac{{{a_i} - 1}}{a_i} \leq \log \left( {{a_i}} \right) \leq {a_i} - 1$.
		Then, we obtain $\frac{{{a_i} - 1}}{K} \leq \log \left( {{a_i}} \right) \leq {a_i} - 1$.
		
		To prove \eqref{eq:intermediate-result-2}, we need to prove two implications: (1) if $\sum_{i=1}^{\infty} \log(a_i) = -\infty$ then $\sum_{i=1}^{\infty} (a_i-1) = -\infty$; and (2) if $\sum_{i=1}^{\infty} (a_i-1) = -\infty$ then $\sum_{i=1}^{\infty} \log(a_i) = -\infty$.
		
		First, let us assume that $\sum_{i=1}^{\infty} \log(a_i) = -\infty$ holds. From the inequality $\frac{a_i-1}{K} \leq \log(a_i)$, we have $\sum_{i=1}^{\infty} (a_i-1) \leq K\sum_{i=1}^{\infty} \log(a_i) = -\infty$.
		Thus, $\sum_{i=1}^{\infty} \log(a_i) = -\infty \implies \sum_{i=1}^{\infty} (a_i-1) = -\infty$.
		Then, let us assume that $\sum_{i=1}^{\infty} (a_i-1) = -\infty$. 
		From the inequality $\log(a_i) \leq a_i - 1$, we have $\sum_{i=1}^{\infty} \log(a_i) \leq \sum_{i=1}^{\infty} (a_i-1) = -\infty$.
		Thus, $\sum_{i=1}^{\infty} (a_i-1) = -\infty \implies \sum_{i=1}^{\infty} \log(a_i) = -\infty$.
		
	Based on these results, we conclude that \eqref{eq:intermediate-result-2} (and hence \eqref{eq:intermediate-result}) holds for when $\{a_i\}$ converges to $M = 1$. This completes the proof of \eqref{eq:intermediate-result}.
	
	Let us rearrange the indices of the interferers $i$ such that they are ordered from the closest to the furthest from the typical user. Then, under the SIR-based coverage criterion, the CCP in \eqref{eq:covcon-SIR} can be written as follows:
	\begin{equation}
		\prod_{i\in\PhiI} \left(1 + \gamma \left(\frac{r_0}{r_i}\right)^{\alpha}\right)^{-1} = \prod_{n = 1}^{N} \left(1 + \gamma \left(\frac{r_0}{r_n}\right)^{\alpha}\right)^{-1}.
	\end{equation}
	We assume that no two interferers have exactly the same distance from the typical user, i.e., $r_1 < r_2 < \ldots< r_N$. Denote $a_n = \left(1 + \gamma \left(\frac{r_0}{r_n}\right)^{\alpha}\right)^{-1}$.
	Thus, $\{a_n\}$ is an ordered sequence where $0 < a_1 < a_2 < \ldots  < 1$.
	Clearly, $\pcovcon > 0$ if $N < \infty$. Assume $N=\infty$. According to \eqref{eq:intermediate-result}, for $\pcovcon$ to be non-zero, it must hold that $\sum_{n=0}^{\infty} (1-a_n) < \infty$. Denote $b_n = {\gamma \left(\frac{r_0}{r_n}\right)^{\alpha}}/{\left(1 + \gamma \left(\frac{r_0}{r_n}\right)^{\alpha}\right)}$. Then, we have the following:
	\begin{align}
		\sum_{n=0}^{\infty} (1-a_n)
		= 
		\sum_{n=0}^{\infty} \frac{\gamma \left({r_0}/{r_n}\right)^{\alpha}}{\left(1 + \gamma \left({r_0}/{r_n}\right)^{\alpha}\right)}, \qquad
		\frac{b_{n+1}}{b_n}
		= 
		\frac{(r_n)^\alpha + \gamma(r_0)^\alpha}{(r_{n+1})^\alpha + \gamma(r_0)^\alpha}
	\end{align}
	Since $r_n < r_{n+1}$ for all $n$, we have $\lim_{n\to\infty} \Bigl|\frac{b_{n+1}}{b_n}\Bigr| < 1$.
	Therefore, according to the convergence test, we obtain $\sum_{n=0}^{\infty} (1-a_n) < \infty$. This completes the proof.
		
		\section{Proof of Proposition \ref{proposition:limit-of-F1-general-case}}\label{appendix:limit-of-F1-general-case}
		
		Define $Z = \pr_h\left[\Delta_{\Phi} >\tau~|~\Phi\right]  = (1-\pcovcon)^\tau$. 
		We can compute $\lim_{\tau\to\infty} Z 
		= \lim_{\tau\to\infty} (1-\pcovcon)^\tau = \mathbbm{1}(\pcovcon = 0)$.
		Since $Z$ only takes values between 0 and 1, the $F_1$ distribution is given as $F_1(\tau) = \e\left[Z\right] = \int_0^1 \pr\left[Z>z\right]dz$
		and hence the limit of $F_1(\tau)$ as $\tau \to \infty$ is computed as:
		\begin{align}
			\lim_{\tau\to\infty} F_1(\tau)
			&= \lim_{\tau\to\infty} \int_0^1 \pr\left[Z > z\right]dz \\
			&\mystepA \int_0^1 \pr\left[1>z\right]\pr\left[\pcovcon = 0\right]dz  + \int_0^1 \pr\left[0>z\right]\pr\left[\pcovcon > 0\right]dz \nonumber
		\end{align}
		where $(a)$ is obtained by separating the integration into two different cases: $\pcovcon = 0$ and $\pcovcon > 0$.
		By using the fact that $\pr\left[0>z\right]=0$ and $\pr\left[1>z\right]=1$ for all $z\in(0,1)$, we have $P_e = \lim_{\tau\to\infty} F_1(\tau) = \pr\left[\pcovcon = 0\right]$. Also, from Lemma \ref{lemma:coverage-probability-general-formula}, we have:
		\begin{align}
			P_e 
			&= \pr\left[\mathcal{G}(r_0)\prod_{i\in \PhiI} \left(1+\gamma\left(\frac{r_0}{r_i}\right)^\alpha\right)^{-1} = 0\right].
		\end{align}
		The proof follows from the fact that, according to Lemma \ref{lemma:positivity-of-pcovcon}, the product inside the probability operator is always greater than zero.

	\end{appendices}
	
	\bibliography{Draft_TWC_Delay} 

\begin{thebibliography}{10}

\bibitem{3gpp2016study}
G.~T. 38.913, ``Study on scenarios and requirements for next generation access
  technologies,'' 2016.

\bibitem{di2018system}
M.~Di~Renzo, A.~Zappone, T.~T. Lam, and M.~Debbah, ``System-level modeling and
  optimization of the energy efficiency in cellular networks - {A} stochastic
  geometry framework,'' {\em IEEE Trans. Wireless Commun.}, vol.~17,
  pp.~2539--2556, 2018.

\bibitem{haenggi2009stochastic}
M.~Haenggi, J.~G. Andrews, F.~Baccelli, O.~Dousse, and M.~Franceschetti,
  ``Stochastic geometry and random graphs for the analysis and design of
  wireless networks,'' {\em IEEE J. Sel. Areas Commun.}, vol.~27, no.~7,
  pp.~1029--1046, 2009.

\bibitem{andrews2011tractable}
J.~G. Andrews, F.~Baccelli, and R.~K. Ganti, ``A tractable approach to coverage
  and rate in cellular networks,'' {\em IEEE Trans. Commun.}, vol.~59, no.~11,
  pp.~3122--3134, 2011.

\bibitem{hmamouche2021new}
Y.~Hmamouche, M.~Benjillali, S.~Saoudi, H.~Yanikomeroglu, and M.~D. Renzo,
  ``New trends in stochastic geometry for wireless networks: A tutorial and
  survey,'' {\em Proceedings of the IEEE}, vol.~109, no.~7, pp.~1200--1252,
  2021.

\bibitem{baccelli2010new}
F.~Baccelli and B.~Blaszczyszyn, ``A new phase transitions for local delays in
  manets,'' in {\em Proc. IEEE INFOCOM}, 2010.

\bibitem{haenggi2012local}
M.~Haenggi, ``The local delay in poisson networks,'' {\em IEEE Trans. Inform.
  Theory}, vol.~59, no.~3, pp.~1788--1802, 2012.

\bibitem{gong2013local}
Z.~Gong and M.~Haenggi, ``The local delay in mobile poisson networks,'' {\em
  IEEE Trans. Wireless Commun.}, vol.~12, no.~9, pp.~4766--4777, 2013.

\bibitem{nie2015hetnets}
W.~Nie, Y.~Zhong, F.-C. Zheng, W.~Zhang, and T.~O'Farrell, ``Hetnets with
  random dtx scheme: Local delay and energy efficiency,'' {\em IEEE Trans. Veh.
  Technol.}, vol.~65, no.~8, pp.~6601--6613, 2015.

\bibitem{zhang2015delay}
G.~Zhang, T.~Q. Quek, A.~Huang, and H.~Shan, ``Delay and reliability tradeoffs
  in heterogeneous cellular networks,'' {\em IEEE Trans. Wireless Commun.},
  vol.~15, no.~2, pp.~1101--1113, 2015.

\bibitem{haenggi2015meta}
M.~Haenggi, ``The meta distribution of the sir in poisson bipolar and cellular
  networks,'' {\em IEEE Trans. Wireless Commun.}, vol.~15, no.~4,
  pp.~2577--2589, 2015.

\bibitem{yang2019meta}
H.~H. Yang and T.~Q. Quek, ``The meta distribution of sinr for small cell
  networks with temporal traffic,'' in {\em Proc. IEEE ICC}, pp.~1--6, IEEE,
  2019.

\bibitem{salehi2017analysis}
M.~Salehi, A.~Mohammadi, and M.~Haenggi, ``Analysis of {D2D} underlaid cellular
  networks: {SIR} meta distribution and mean local delay,'' {\em IEEE Trans.
  Commun.}, vol.~65, no.~7, pp.~2904--2916, 2017.

\bibitem{deng2018meta}
N.~Deng and M.~Haenggi, ``The meta distribution of the {SINR} and rate in
  heterogeneous cellular networks,'' in {\em 2018 IEEE PIMRC}, pp.~1--6, IEEE,
  2018.

\bibitem{BaccelliBook2009}
B.~Blaszczyszyn, {\em Stochastic Geometry and Wireless Networks, Part I:
  Theory}.
\newblock Now Publishers Inc.

\bibitem{yu2013downlink}
S.~M. Yu and S.-L. Kim, ``Downlink capacity and base station density in
  cellular networks,'' in {\em Proc. WiOpt)}, 2013.

\bibitem{di2016intensity}
M.~Di~Renzo, W.~Lu, and P.~Guan, ``The intensity matching approach: A tractable
  stochastic geometry approximation to system-level analysis of cellular
  networks,'' {\em IEEE Trans. Wireless Commun.}, vol.~15, no.~9,
  pp.~5963--5983, 2016.

\bibitem{bennis2018ultrareliable}
M.~Bennis, M.~Debbah, and H.~V. Poor, ``Ultrareliable and low-latency wireless
  communication: Tail, risk, and scale,'' {\em Proc. IEEE}, vol.~106, no.~10,
  pp.~1834--1853, 2018.

\bibitem{chiu2013stochastic}
S.~N. Chiu, D.~Stoyan, W.~S. Kendall, and J.~Mecke, {\em Stochastic geometry
  and its applications}.
\newblock John Wiley \& Sons, 2013.

\bibitem{gil1951note}
J.~Gil-Pelaez, ``Note on the inversion theorem,'' {\em Biometrika}, vol.~38,
  no.~3-4, pp.~481--482, 1951.

\bibitem{abate1995numerical}
J.~Abate and W.~Whitt, ``Numerical inversion of laplace transforms of
  probability distributions,'' {\em ORSA J. Computing}, vol.~7, no.~1,
  pp.~36--43, 1995.

\bibitem{ko2000outage}
Y.-C. Ko, M.-S. Alouini, and M.~K. Simon, ``Outage probability of diversity
  systems over generalized fading channels,'' {\em IEEE Trans. Commun.},
  vol.~48, no.~11, pp.~1783--1787, 2000.

\bibitem{di2014stochastic}
M.~Di~Renzo and P.~Guan, ``Stochastic geometry modeling of coverage and rate of
  cellular networks using the gil-pelaez inversion theorem,'' {\em IEEE Commun.
  Lett.}, vol.~18, no.~9, pp.~1575--1578, 2014.

\bibitem{wang2019meta}
S.~Wang and M.~Di~Renzo, ``On the meta distribution in spatially correlated
  non-poisson cellular networks,'' {\em EURASIP J. Wireless Commun. Netw.},
  vol.~2019, no.~1, p.~161, 2019.

\bibitem{haenggi2018efficient}
M.~Haenggi, ``Efficient calculation of meta distributions and the performance
  of user percentiles,'' {\em IEEE Wireless Commun. Lett.}, vol.~7, no.~6,
  pp.~982--985, 2018.

\bibitem{topsok2006some}
F.~Topsok, ``Some bounds for the logarithmic function,'' {\em Inequality Theory
  and Applications}, vol.~4, p.~137, 2006.

\end{thebibliography}
	\bibliographystyle{ieeetr}

\end{document}